\documentclass{LMCS}

\def\dOi{9(4:22)2013}
\lmcsheading%
{\dOi}
{1--39}
{}
{}
{Jan.~13, 2013}
{Dec.~17, 2013}
{}

\ACMCCS{[{\bf Security and privacy}]: Formal methods and theory of
  security---Logic and verification; [{\bf Theory of computation}]:
  Formal languages and automata theory---Automata
  extensions---Transducers\,/\,Quantitative automata;
  Logic---Verification by model checking}


\pdfoutput=1
\usepackage[utf8x]{inputenc}
\usepackage{amsmath}
\usepackage{amssymb}
\usepackage{stmaryrd}
\usepackage{bbm}
\usepackage{xspace}
\usepackage{mathrsfs}
\usepackage[ruled,commentsnumbered,linesnumbered,norelsize]{algorithm2e}
\usepackage{mathtools}
\usepackage{hyperref}

\usepackage{tikz}
\usetikzlibrary{calc,matrix,chains,positioning,arrows,fit,automata,backgrounds}
\usetikzlibrary{shapes,shapes.multipart,decorations,decorations.pathmorphing,decorations.pathreplacing}
\tikzstyle{dotstyle}=[fill=black,circle,minimum size=3pt,inner sep=0pt]
\tikzstyle{automatapath}=[->,decorate,decoration={snake,segment length=5mm}]
\tikzstyle{curlybracket}=[decorate,decoration={brace,amplitude=7pt}]

\allowdisplaybreaks[1]

\newcommand{\apath}[2][]{\xrightarrow[#1]{#2}}
\newcommand{\apathS}[2][]{\xrightarrow[#1]{#2}\!\!{}^*\,}
\newcommand{\nat}{\mathbbm{N}}
\newcommand{\eps}{\varepsilon}

\newcommand{\DeltaR}{\Delta_{\mathfrak R}}
\DeclareMathOperator{\dotcup}{\mathaccent\cdot\cup}

\newcommand{\FO}{\ensuremath{\mathsf{FO}}}

\newcommand{\pad}{\square}

\newcommand{\pathTo}[1][]{\xrightarrow{#1}^*}

\newcommand{\defEquiv}{:\Leftrightarrow}
\newcommand{\Equiv}{\Leftrightarrow}
\newcommand{\Implies}{\Rightarrow}

\newcommand{\configstep}[1][]{\vdash_{#1}}

\newcommand{\natInf}{\nat \cup \{ \infty \}}

\newcommand{\automatonA}{\mathfrak A}
\newcommand{\automatonB}{\mathfrak B}
\newcommand{\automatonC}{\mathfrak C}

\newcommand{\automatonT}{\mathfrak T}

\newcommand{\In}{\mathsf{In}}
\newcommand{\Fin}{\mathsf{Fin}}

\newcommand{\semantics}[1]{\llbracket#1\rrbracket}

\newcommand{\costlea}[1][\alpha]{\preceq_{#1}}
\newcommand{\costEquiv}[1][\alpha]{\approx_{#1}}

\newcommand{\infProjection}[1]{_{\inf,#1}}
\newcommand{\supProjection}[1]{_{\sup,#1}}

\newcommand{\ihle}{\stackrel{\scriptstyle{\mathclap{\mathrm{i.h.}}}}{\le}}

\makeatletter
\newcommand{\commentLe}{\@ifstar
  \commentLeStar%
  \commentLeNoStar}
\newcommand{\commentGe}{\@ifstar
  \commentGeStar%
  \commentGeNoStar}
\newcommand{\commentEq}{\@ifstar
  \commentEqStar%
  \commentEqNoStar}
\newcommand{\commentEquiv}{\@ifstar
  \commentEquivStar%
  \commentEquivNoStar}
\newcommand{\commentOp}{\@ifstar
  \commentOpStar%
  \commentOpNoStar}

\newcommand{\commentLeNoStar}[1]{\stackrel{\scriptstyle{\mathclap{#1}}}{\le}}
\newcommand{\commentGeNoStar}[1]{\stackrel{\scriptstyle{\mathclap{#1}}}{\ge}}
\newcommand{\commentEqNoStar}[1]{\stackrel{\scriptstyle{\mathclap{#1}}}{=}}
\newcommand{\commentEquivNoStar}[1]{\stackrel{\scriptstyle{\mathclap{#1}}}{\Equiv}}
\newcommand{\commentLeStar}[1]{\stackrel{\scriptstyle{#1}}{\le}}
\newcommand{\commentGeStar}[1]{\stackrel{\scriptstyle{#1}}{\ge}}
\newcommand{\commentEqStar}[1]{\stackrel{\scriptstyle{#1}}{=}}
\newcommand{\commentEquivStar}[1]{\stackrel{\scriptstyle{#1}}{\Equiv}}

\newcommand{\commentOpStar}[2]{\stackrel{\scriptstyle{#1}}{#2}}
\newcommand{\commentOpNoStar}[2]{\stackrel{\scriptstyle{\mathclap{#1}}}{#2}}

\makeatother

\newcommand{\cprofile}{\mathsf{Profile}}

\newcommand{\counterProfiles}{\mathcal{CP}}

\newcommand{\cple}{\le_{\mathrm{cw}}}

\newcommand{\ipl}{i^{+}_{\leftarrow}}
\newcommand{\ipr}{i^{+}_{\rightarrow}}

\newcommand{\cmax}{c_{\mathit{max}}}

\newcommand{\na}{\diagup}

\newenvironment{indt}{\begin{quote}}{\end{quote}}

\newcommand{\prsR}{\mathfrak R}
\newcommand{\configstepCost}[1]{\vdash_{#1}}
\newcommand{\configstepsCost}[2][*]{\vdash^{#1}_{#2}}
\newcommand{\configstepsLesserCost}[2][*]{\vdash^{#1}_{\le #2}}
\newcommand{\configstepsLesserCostLangRestrict}[2]{\vdash^{#1}_{\le
#2}}

\newcommand{\pdsrule}[3]{#1 \xrightarrow[#2]{} #3}

\newcommand{\padprodL}{\otimes_{\mathsf{L}}}
\newcommand{\padprodR}{\otimes_{\mathsf{R}}}
\newcommand{\alphVector}[2]{#1^{\otimes #2}}
\newcommand{\alphVectorS}[2]{\left(#1^{\otimes #2}\right)^*}
\newcommand{\alphVectorL}[2]{#1^{{\padprodL #2}^*}}

\newcommand{\NalphVectorL}[2]{\overline{#1^{{\padprodL #2}^*}}}

\newcommand{\rl}{m_{\leftarrow}}
\newcommand{\rr}{m_{\rightarrow}}
\newcommand{\hrl}{\hat m_{\leftarrow}}
\newcommand{\hrr}{\hat m_{\rightarrow}}
\newcommand{\rev}{\mathsf{rev}}

\newcommand{\structureS}{\mathfrak S}

\newcommand{\free}{\mathsf{free}}
\newcommand{\FORR}{\ensuremath{\mathsf{FO\!\!+\!\!RR}}}

\newcommand{\valfunc}{\mathfrak I}

\newcommand{\bindL}{\mathsf{conv}}

\newcommand{\unbind}{\mathsf{unconv}}

\newcommand{\tspathTo}{\rightarrowtail^*}

\newcommand{\iOp}{\ensuremath{\mathtt{i}}}
\newcommand{\rOp}{\ensuremath{\mathtt{r}}}
\newcommand{\nOp}{\ensuremath{\mathtt{n}}}
\newcommand{\cOp}{\ensuremath{\mathtt{c}}}
\newcommand{\icOp}{\ensuremath{\mathtt{i\! c}}}
\newcommand{\crOp}{\ensuremath{\mathtt{c\! r}}}

\newcommand{\twosm}[2]{\begin{smallmatrix} #1 \\ #2 \end{smallmatrix}}

\newcommand{\inductionstart}[1]{\begin{description}\item[\textbf{(base
case)}]#1 \\}
\newcommand{\inductionstep}[1]{\item[\textbf{(induction step)}]
#1 \\}
\newcommand{\inductionend}{\end{description}}

\newcommand{\profileToOp}{\mathsf{PrToOp}}

\newcommand{\fraku}{\mathfrak u}

\newcommand{\RPRS}{\textsf{RPRS}}
\newcommand{\resRep}[1]{\mathcal C_{#1}}
\newcommand{\annotationMonoidM}{\mathcal M}
\newcommand{\neutralM}[1][]{e_{\annotationMonoidM_{#1}}}

\newcommand{\wsel}{\mathsf{wsel}}
\newcommand{\strip}{\mathsf{strip}}

\renewcommand{\theta}{\vartheta}
\newcommand{\twovec}[2]{\begin{pmatrix}#1 \\ #2\end{pmatrix}}

\begin{document}
\title{Modeling and Verification of Infinite Systems with Resources}

\author[M.~Lang]{Martin Lang}
\address{Chair of Computer Science 7, RWTH Aachen, 52056 Aachen (Germany)}
\email{\{lang,loeding\}@automata.rwth-aachen.de}

\author[C.~Löding]{Christof L{\"o}ding}

\begin{abstract}
  We consider formal verification of recursive programs with resource
consumption. We introduce prefix replacement systems with non-negative integer
counters which can be incremented and reset to zero as a formal model for such
programs. In these systems, we investigate bounds on the resource consumption
for reachability questions. Motivated by this question, we introduce relational
structures with resources and a quantitative first-order logic over these
structures. We define resource automatic structures as a subclass of these
structures and provide an effective method to compute the semantics of the logic
on this subclass. Subsequently, we use this framework to solve the bounded
reachability problem for resource prefix replacement systems. We achieve this 
result by extending the well-known saturation method to annotated prefix replacement
systems. Finally, we provide a connection to the study of the logic cost-WMSO.
\end{abstract}

\maketitle

\section{Introduction}

Transition systems induced by the configuration graph of pushdown automata have
become an important tool in automatic program verification. Starting with the
introduction of the concept of pushdown automata by A.G.\ Oettinger in 1961 and
M.-P.\ Sch\"{u}tzenberger in 1963, these systems have been extensively studied
and are well understood today. Already in 1985, Muller and Schupp were able to
prove the decidability of MSO-logic over these systems (see
\cite{theoryOfPushdownAutomata}), which are mostly called pushdown systems
nowadays. Until now, methods and algorithms have been developed that provide
automatic verification procedures for recursive programs. For example, the model-checker
jMoped, introduced in \cite{jMoped}, uses symbolic pushdown systems to verify
programs given in Java bytecode. 

Recently, quantitative aspects in formal verification came into the focus.
Although there is a long history of automata models with quantitative aspects
such as weighted automata (see \cite{Schutzenberger-WeightedAutomata}) or
distance automata (see \cite{distanceAutomata}), their use in the area of formal
verification was limited. More recently, timed automata (see
\cite{Alur-timedAutomata}) were introduced as model for systems with
quantitative timing constraints. Additionally, there is currently some effort to
extend this line of research to the area of pushdown systems. In
\cite{timed-pushdown}, a model that combines dense-timed automata with pushdown
automata is introduced in order to model real-time recursive systems. Furthermore,
finite automata and games with resource constraints instead of timing
constraints were introduced and studied (see
\cite{infinite-runs-in-weighted-timed-automata}). Additionally, a combined
concept of weighted automata and pushdown systems is used for program
verification with data flow analysis (see \cite{Lal-TACAS08}). However, the
authors do not know of research considering infinite systems with resource
constraints.

We contribute in our work to the theory of quantitative systems by defining and
analyzing a model for recursive programs with discrete resource consumption. We
introduce resource prefix replacement systems as a combination of prefix
replacement systems and non-negative integer counters similar to those used in B-automata (see
\cite{regularcostfunctions}). These counters support three kinds of operations:
increment, reset to zero, or skip. The operations annotate the
transitions of the prefix replacement system. Thereby, the model is able to
simulate usage and refreshment of resources during program execution. We
consider the resource usage of a run through the system to be the highest
occurring counter value.

One central aspect in systems with resource consumption is boundedness.
Generally, this means checking whether there is a finite bound such that the
system can complete a given task while keeping the resource consumption within
the bound. We formalize this idea by the concept of bounded reachability. This
means checking whether there is a finite bound such that it is possible to reach
from all initial configurations of the system some configuration in a given goal
set with less resource usage than the bound. An algorithmic solution to this
problem can be used as a building block in the realizability checking of systems
with resources. For example, consider a battery powered mobile measuring device
which should be able to complete certain tasks without recharging the battery.
This is only possible if there is a finite bound on the energy consumption
independent of the selected task. The realizability question of this requirement
can be stated in the form of a bounded reachability problem.

Motivated by this question, we develop a framework to formalize and
solve the bounded reachability problem and related questions for recursive
programs with resource consumption. We introduce resource structures as
a quantitative variant of relational structures based on the idea that being in
relation may cost a certain amount of resources. For these structures, we
develop the quantitative logic first-order+resource relations (for short
\FORR{}) in order to express combined constraints on the resource
consumption and the behavior of the system. Intuitively, the semantics of this
logic is designed to describe the amount of resources necessary to satisfy a
first-order constraint. We define resource automatic structures as a subclass of
resource structures. This definition extends automatic structures as introduced
in \cite{automatic-structures} with a quantitative aspect. Based on the closure
properties of B-automata and ideas from the theory of automatic structures, we
provide an effective way to compute the semantics of the logic over resource
automatic structures.

Subsequently, we demonstrate the usage of this general theory to
solve the bounded reachability problem on resource prefix replacement systems
for regular initial and goal sets. This is achieved by showing that resource
prefix replacement systems are resource automatic structures and thus bounded
reachability can be solved by computing the semantics of an \FORR{} formula. In
order to obtain this decidability result, we analyze prefix replacement systems with a
general form of annotation. Based on the well-known saturation principle, we
devise a method to compute an annotation aware transitive closure of the
successor relation in the form of synchronized transducers. 

This saturation procedure is based on a saturation for ground-tree
transducers presented in \cite{gtt-saturation} and constructed with
the goal of obtaining a synchronous transducer. Although the idea of saturation
in annotated systems is not entirely new, the existing methods do not quite fit
our needs. In \cite{RepsSchwoon-WeightedPushdown}, saturation is used to compute
predecessor and successor configurations of regular configuration sets for
pushdown systems with weights from a semi-ring. This result
was extended in \cite{Lal-TACAS08} to compute an asynchronous transducer for
the reachability relation of (semi-ring) weighted pushdown systems. However,
these methods cannot be applied directly to our problem since we need a 
synchronous transducer for the decision procedure of \FORR{} and additionally
encoding the resource counters into a semi-ring structure is very complex. 
Moreover, we believe that this different two-sided-saturation approach
offers an interesting new viewpoint  on the problem of transitive closure in
annotated pushdown- and prefix replacement systems.

Finally, we provide a connection between our problems and the logic cost-WMSO,
which was also introduced in connection to B-automata. We show that the
decidability results for the boundedness problem of cost-WMSO presented by M.
Vanden Boom in \cite{costwMSO} can also be used to obtain a decision procedure
for a restricted version of bounded reachability. Although this part does not
contain any new results, it shows that this logic, which was designed 
as equivalent formalism to B-automata, is also capable of expressing
a very natural problem for quantitative systems. This motivates further studies
on the expressive power of cost-(W)MSO and other quantitative logics that emerged
around cost-automata.

In the following section, we first introduce and formalize our model for
recursive programs with resource consumption. In
Section~\ref{sec:Preliminaries}, we present the known results for counter
automata as introduced by T. Colcombet and briefly repeat the important
results. Subsequently, we define and investigate the model of resource
structures and the logic \FORR{} in
Section~\ref{sec:ResourceStructuresAndLogic}. Next, we exhibit an extended
saturation approach which enables us to compute the transitive closure
for prefix replacement systems with annotations in
Section~\ref{sec:ReachabilityWithAnnotations}. At the end of this section, we
use the previously developed framework to prove our main statement on the
bounded reachability problem. Section~\ref{sec:BoundedReachabilityAndCostWMSO}
connects our results to the logic cost-WMSO. Finally, we conclude in
Section~\ref{sec:Conclusion}.

We would like to thank T.~Colcombet and M.~Bojańczyk for the fruitful and
enlightening discussions. Moreover, we want to thank the anonymous reviewers 
for their very constructive and detailed comments. They helped us to 
significantly improve the quality of this article. Additionally, we thank the 
Deutsche Forschungsgemeinschaft, which supports the first author in 
the project ``Automatentheoretische Verifikationsprobleme mit Ressourcenschranken''.

\section{Resource Prefix Replacement Systems}
\label{sec:ResourcePRS}

We define \emph{resource prefix replacement systems} (for short \RPRS{}) to be a
suitable model for recursive programs with resource consumption. This is
achieved by combining prefix replacement systems, which are a well-known model
for recursive programs, with discrete non-negative counters. These counters
support three kinds of operations, which annotate the transitions of the
prefix replacement system. First, the counter can be incremented (\texttt{i}).
This models the use of one unit of one resource. Second, the counter can be
reset to zero (\texttt{r}). This models the complete refresh/refill of this
kind of resources. Third, it is possible to leave the counter unchanged, what we
call no operation (\texttt{n}). This counter model is similar to B-automata,
which are presented in the subsequent section. We formalize the model of
\RPRS{} by the following definition.

\begin{defi}
  A resource prefix replacement system is a triple $\prsR =
(\Sigma,\Delta,\Gamma)$ consisting of a finite  alphabet $\Sigma$, a finite set
of counters $\Gamma$ and a finite transition relation $\Delta \subseteq \Sigma^+
\times \Sigma^* \times \{\nOp,\iOp,\rOp\}^\Gamma$. A configuration is a finite
word over $\Sigma$. A transition $(u,v,f) \in \Delta$ enables a change from the
configuration $uw$ to $vw$ for all $w \in \Sigma^*$ and triggers the counter
operation $f(c)$ for each counter $c \in \Gamma$. If there is only one counter
in the system, we also write $u \apath[\mathtt{op}]{} v$ for a rule
$(u,v,f)$ where $f(c_0) = \mathtt{op}$ for the unique counter $c_0$.
\end{defi}

We are mainly interested in the configuration graph induced by the system. A
path in this graph represents a (partial) run of the modeled recursive program.
The resource usage of such a run is calculated by simulating all counter
operations along the path according to the annotated operations at the
transitions.  We identify the resource usage (or resource-cost) of the path with
the maximal counter value (of all counters) in the sequence. Furthermore, we
write $u \configstepsLesserCost{k} v$ if $v$ is reachable from $u$ with resource
usage of at most $k$. We remark that we do not distinguish between the different
counters when calculating the overall resource usage because we focus on
boundedness questions. 


For example, consider a simple \RPRS{} with only one counter $c_0$ over the
alphabet $\{a,b\}$. The system contains four replacement rules: $a
\apath[\iOp]{} \eps , a \apath[\rOp]{} ba, b \apath[\rOp]{} bb, b \apath[\nOp]{}
a$. Figure~\ref{fig:RPRSExample} shows a part of the resulting configuration
graph. In this example, it can be seen that for all $n \in \nat$, we obtain $a^n
\configstepsLesserCost{2} \eps$. However, this is not possible with simple paths
without loops because detours using the replacement rules annotated with reset
are needed. In detail, we obtain, e.g., $aaa \configstepsLesserCost{2} \eps$ by
$aaa \apath[\iOp]{} aa \apath[\iOp]{} a \apath[\rOp]{} ba \apath[\nOp]{} aa
\apath[\iOp]{} a \apath[\iOp]{} \eps$.

\begin{figure}[t]
  \begin{center}
	\begin{tikzpicture}
	\begin{scope}[level distance=2.2cm,
				level 1/.style={sibling distance=3cm},
				level 2/.style={sibling distance=1.5cm},
				level 3/.style={sibling distance=0.8cm},
				level 4/.style={level distance=0.7cm},
				edge from parent/.style={draw=none},
				font={\large},
				scale=0.7,
				transform shape]

	\tikzstyle{littleN}=[midway,right,font={\small}]
	\tikzstyle{nodehighlight}=[color=blue!70!white]
	\tikzstyle{pathhighlight}=[color=red!70!white,dashed]
	\tikzstyle{counterAnnotation}=[above left=1mm,font={\small},color=red!80]

	\node (eps) at (0,0) {$\eps$} [grow'=right]
		child {
		node (a) {$a$}
		child {
			node (aa) {$aa$}
			child {node (aaa) {$aaa$}
			child {node {$\hdots$}}}
			child {node (baa) {$baa$}
			child {node {$\hdots$}}}
		}
		child {
			node (ba) {$ba$}
			child {node (aba) {$aba$}
			child {node {$\hdots$}}}
			child {node (bba) {$bba$}
			child {node {$\hdots$}}}
		}
		}
		child {
		node (b) {$b$}
		child {
			node (ab) {$ab$}
			child {node (aab) {$aab$}
			child {node {$\hdots$}}}
			child {node (bab) {$bab$}
			child {node {$\hdots$}}}
		}
		child {
			node (bb) {$bb$}
			child {node (abb) {$abb$}
			child {node {$\hdots$}}}
			child {node (bbb) {$bbb$}
			child {node {$\hdots$}}}
		}
		};

	\draw[->] (aaa) -- (aa) node[midway,above] {$\iOp$};
	\draw[->] (aba) -- (ba) node[midway,above] {$\iOp$};
	\draw[->] (aab) -- (ab) node[midway,above] {$\iOp$};
	\draw[->] (abb) -- (bb) node[midway,above] {$\iOp$};
	\draw[->] (aa) -- (a) node[midway,above] {$\iOp$};
	\draw[->] (ab) -- (b) node[midway,above] {$\iOp$};
	\draw[->] (a) -- (eps) node[midway,above] {$\iOp$};

	\draw[<-] (bbb) -- (bb) node[midway,below] {$\rOp$};
	\draw[<-] (bba) -- (ba) node[midway,below] {$\rOp$};
	\draw[<-] (bab) -- (ab) node[midway,below] {$\rOp$};
	\draw[<-] (baa) -- (aa) node[midway,below] {$\rOp$};
	\draw[<-] (bb) -- (b) node[midway,below] {$\rOp$};
	\draw[<-] (ba) -- (a) node[midway,below] {$\rOp$};

	\draw[<-] (aaa) -- (baa) node[littleN] {$\nOp$};
	\draw[<-] (aba) -- (bba) node[littleN] {$\nOp$};
	\draw[<-] (abb) -- (bbb) node[littleN] {$\nOp$};
	\draw[<-] (aab) -- (bab) node[littleN] {$\nOp$};
	\draw[<-] (aa) -- (ba) node[midway,right] {$\nOp$};
	\draw[<-] (ab) -- (bb) node[midway,right] {$\nOp$};
	\draw[<-] (a) -- (b) node[midway,right] {$\nOp$};
	\end{scope}
\end{tikzpicture}
  \end{center}
  \caption{Configuration graph induced by the example \RPRS{}}
  \label{fig:RPRSExample}
\end{figure}
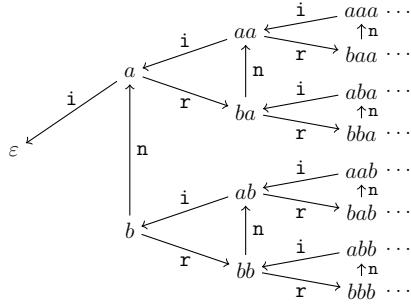

\subsection{The Bounded Reachability Problem}
\label{subsec:BoundedReachability}

Reachability checking is a fundamental building block of many formal
verification procedures. For example, it can be applied as a decision procedure
for the termination problem of a program. The general reachability problem 
on prefix replacement systems can be formulated as follows. Let $A$ (starting 
configurations) and $B$ (final configurations) be two sets of
configurations. We say $B$ is reachable from $A$ if for all elements
in $A$ there is a path to some element of $B$. This resembles the termination
idea that independent of the starting configuration of the program a final
configuration should be reached after finitely many steps.

However, in the context of systems with resources mere reachability is
often not enough to ensure the realizability of the system. 
We therefore extend the question by asking whether bounded
resources are enough to achieve reachability.

\begin{defi}[Bounded Reachability]
  Let $A,B$ be two sets of configurations. We say $B$ is boundedly reachable
  from $A$ if there is a bound $k \in \nat$ such that for all $u \in A$ there
  is a $v \in B$ satisfying $u \configstepsLesserCost{k} v$.
\end{defi}

Reconsider the example depicted in Figure \ref{fig:RPRSExample}. We already
observed
that $a^n \configstepsLesserCost{2} \eps$. Consequently, the set $\{\eps\}$ is
boundedly reachable from $\{ a^n \mid n \in \nat\}$. However, if we remove one
of the rules  $a \apath[\rOp]{} ba, b \apath[\nOp]{} a$, this does not
hold anymore although it is still possible to reach $\eps$ but with increasing
resource usage.

The rest of this work is dedicated to developing tools with the goal of better
understanding and solving the bounded reachability problem. First, we
introduce the basic concepts and known results on cost automata in
Section~\ref{sec:Preliminaries}. These results form the basis for most of our
following work. In Section~\ref{sec:ResourceStructuresAndLogic}, we define
resource structures as a specialized concept to represent systems with resource
consumption. On these structures we define the logic \FORR{} to describe
the behavior of the system in combination with its resource consumption. This
framework allows amongst other things a simple formulation of the bounded reachability problem. To
achieve the desired goal and solve the bounded reachability problem in a
restricted scenario, we provide a method to effectively compute the
semantics of \FORR{} on a restricted class of structures which we call resource
automatic. In Section~\ref{sec:ReachabilityWithAnnotations}, we develop a method
to compute an annotation aware transitive closure of annotated prefix
replacement systems. With this method we are able to prove that \RPRS{} are
resource automatic. Thus, we solve the bounded reachability problem for regular
sets of configurations.

\section{Cost Automata}
\label{sec:Preliminaries}


To study boundedness questions, M.~Bojańczyk and T.~Colcombet
presented general automata with (non-readable) counters called B- and
S-automata in \cite{bounds-in-omega-regularity}. T.~Colcombet compared these
models with an algebraic and a logical approach in \cite{regularcostfunctions}.
All these concepts define functions from $\Sigma^*$ to $\natInf$. In 
\cite{regularcostfunctions} it is shown that all the models are capable of 
expressing the same functions up to some equivalence relation $\costEquiv[]$. 
The equivalence classes resulting from this relation are called 
\emph{cost functions}. We additionally call a cost function \emph{regular} if
it contains a function that is definable by some B-automaton. Correspondingly,
we call the two automata models \emph{cost automata}.

The definition of the equivalence relation $\costEquiv[]$  uses so called
\emph{correction functions} to measure the difference between a pair of cost
functions. A correction function $\alpha$ is a non-decreasing mapping from
$\natInf$ to $\natInf$ which maps $\infty$ and only $\infty$ always to $\infty$,
i.e., $k \le j \Implies \alpha(k) \le \alpha(j)$ and $\alpha(x) = \infty \Equiv
x = \infty$.
\begin{defi}
	Let $x,y \in \natInf$ be two values and $\alpha: \natInf
	\to \natInf$ a correction function. We say that $x$ is $\alpha$-dominated
	by $y$ and write $x \costlea y$ if $x \le \alpha(y)$. If $x \costlea y$ 
	and $y \costlea x$, we say that $x$ and $y$ are $\alpha$-equivalent and
	write $x \costEquiv y$.

	We extend this naturally  to functions. Let $f,g: \Sigma^* \to \natInf$ be 
	two functions. We say $f$ is $\alpha$-dominated by $g$ and write 
	$f \costlea g$ if 
	\[ \forall x \in \Sigma^*: f(x) \costlea g(x) \]
	Analogously, if $f \costlea g$ and $g \costlea f$ we say $f$ and $g$ are
	$\alpha$-equivalent and write $f \costEquiv g$. The two functions are just
	called equivalent (written $f \costEquiv[] g$) if there exists some
	correction function $\alpha$ such that $f \costEquiv g$.
\end{defi}
Note that for a fixed $\alpha$, the relation $\costEquiv$ is not transitive.
From the inequality it becomes clear that for three functions $f \costEquiv g
\costEquiv[\beta] h$, we only obtain $f \costlea[\alpha \circ \beta] h$ and $h
\costlea[\beta \circ \alpha] f$. However, one can easily check that for $\gamma
:= \max(\beta \circ \alpha,\alpha \circ\beta)$, we obtain $f \costEquiv[\gamma]
h$. 

An equivalent characterization of the relation $\costEquiv[]$ can be given by
comparing the subsets of the domain with bounded image. 

\begin{lem}[see \cite{regularcostfunctions}]
\label{lem:SecondCostFunctionDefinition}
	Let $f$, $g$ be two cost functions and $A \subseteq \Sigma^*$. We write
$f(A) < \infty$ if $f$ is bounded on $A$, i.e., if $\sup_{a \in A} f(x) <
\infty$. 

We have $f \costEquiv[] g$ iff for all sets $A \subseteq \Sigma^*$: $f(A) <
\infty \Equiv g(A) < \infty$. 
\end{lem}

The definition of $\costEquiv[]$ with correction function has the
advantage that it is also applicable to single values instead of functions.
Thus, it enables us to prove the equivalence of two functions by comparing all
of their values instead of all subsets of the domain. The second
characterization helps to understand the flexibility of the relation
$\costEquiv[]$ and the expressiveness of regular cost functions. 

In order to
reduce the technical overhead and make full use of the expressiveness of 
B-/S-automata as long as possible, we usually work with concrete cost functions 
defined by automata. We view these functions as representatives of their 
$\costEquiv[]$-equivalence class. As a consequence, we will also consider
concrete values of the functions for a single word although this is not
well defined for the equivalence class of functions. We explicitly mention the
situations in which it is important that functions have to be considered
only up to the equivalence $\costEquiv[]$.

In the following, we introduce the models of B- and S-automata as well as some
known results for these models. The structures of B- and S-automata are identical.
They only differ in their semantics. We first provide the structure and
introduce the semantics later.

\begin{defi}[see \cite{regularcostfunctions}]
A B/S-automaton is a nondeterministic finite automaton (NFA) extended with a
finite set of counters. Formally, we have $\automatonA = (Q, \Sigma, \Delta, \In,\Fin,\Gamma)$.
Similar to standard NFAs, $Q$ is a finite set of states, $\In$ the set of
initial states, $\Fin$ the set of final states and $\Sigma$ a finite input
alphabet. $\Gamma$ is the finite set of counters and the counter operations 
are annotated to the transitions. The finite transition relation has the 
form $\Delta \subseteq Q \times \Sigma \times Q \times (\{\iOp,\rOp,\cOp\}^*)^\Gamma$.
\end{defi}

The counter operations in B-/S-automata are similar to \RPRS{}. However,
we don't have the operation $\nOp$ but an operation $\cOp$ which is explained 
below. Additionally, we call B-automata \emph{simple} if
their transition relation contains only counter operations of the form 
$\eps,\rOp$ and $\icOp$. Similarly, we call S-automata \emph{simple} if their transition
relation contains only counter operations of the form $\eps,\iOp,\rOp$ and $\crOp$.

The semantics of cost automata is based on the notion of a run. A run is 
a sequence $t_1,\ldots,t_{n}$ of transitions that is
compatible with the input word, i.e., for a word $w = a_1\ldots a_n$
with $a_i \in \Sigma$, we have $t_i = (q_{i-1},a_{i},q_{i},\fraku_{i}) \in \Delta$.
Similar to \RPRS{}, we
assign a value to each run of a B- or S-automaton based on the annotated counter
operations. However, the counters of cost automata support the additional $\cOp$
counter operation which means \emph{check}. For the value of the run, only the
counter values of checked positions are considered. The need for this additional
action will become clear after the definition of S-automaton semantics. 
For a run $\rho$ we denote the set of all counter values at checked positions 
with $\mathrm{C}(\rho)$. A run is called \emph{accepting} if its first transition
starts in a state of $\In$ and its last transition ends in a state of $\Fin$.
We formally define the B- and S-semantics for a cost automaton $\automatonA$
as follows:
\begin{align*}
	\semantics{\automatonA}_B(w) &:= \inf_{\rho : \text{ acc. run of
$\automatonA$ on $w$}} \sup \mathrm{C(\rho)} \\
	\semantics{\automatonA}_S(w) &:= \sup_{\rho : \text{ acc. run of
$\automatonA$ on $w$}} \inf \mathrm{C(\rho)} 
\end{align*}
We remind the reader that $\inf \emptyset = \infty$ and $\sup \emptyset
= 0$. We see that the value of an S-run would always be $0$ if we had no check 
operation and considered all occurring values because the counters are 
initialized with $0$. 

Figure~\ref{fig:CounterAutomataExaple} gives an example for this definition. The
shown automata have only one counter $c_0$ and both count the maximal number of
subsequent $a$-letters in a word. On the left-hand side of the figure there is
an automaton with B-semantics which defines this function. The automaton on the
right-hand side uses S-semantics. A large class of examples is formed by the
characteristic functions of regular languages. For a language $L \subseteq
\Sigma^*$, we define the characteristic function $\chi_L: \Sigma^* \to \natInf$
to assume the value $0$ if the word is in the language $L$ and $\infty$ 
otherwise. An NFA $\automatonA$ for $L$ can be transformed into a B-automaton
$\automatonA'$ which defines $\chi_L$ by taking $\automatonA$, adding one 
counter and setting all counter operations of the transitions to $\eps$. Thus,
$\sup C(\rho) = \sup \emptyset$ is $0$ for every run $\rho$. So, the value of 
every word in $L$ is $0$ and the value of all other words is $\infty$ since 
there is no accepting run and consequently we have $\inf \emptyset = \infty$.
It is easy to see that the automaton $\automatonA'$ equipped with S-semantics 
defines the characteristic function $\chi_{\overline{L}}$ of the complement of
$L$. 

Conversely, it is also possible to obtain regular languages from functions
defined by B- or S-automata. We remark that the language of all words which
have a value less than a fixed $k \in \nat$ in the given function is regular.
Since the bound $k$ is fixed, a usual finite automaton can simulate the counter
values up to $k$ and thus decide whether a given word is below or above the
threshold. 

\begin{figure}
	\begin{tikzpicture}[initial text=,>=stealth]
	\begin{scope}
		\node[state,initial,accepting] (q_0) {$q_0$};
		\node[above left=1cm] at (1,1) {B-semantics:};
		\path 	(q_0) edge[loop above] node {$a : \icOp$} ()
				(q_0) edge[loop below] node {$b : \rOp$} ();
	\end{scope}

	\begin{scope}[node distance=2.5cm,on grid,auto,xshift=4cm]
		\node[state,initial] (q_0) {$q_0$};
		\node[above left=1cm] at (1,1) {S-semantics:};
		\node[state,accepting,right=of q_0] (q_1) {$q_1$};
		\node[state,accepting,right=of q_1] (q_2) {$q_2$};
		\draw[->] 	(q_0) edge[loop above] node {$a,b : \eps$} ()
				(q_0) edge[bend right] node[swap] {$b: \crOp $} (q_2)
				(q_0) edge node {$a : \iOp$} (q_1)
				(q_1) edge[loop above] node {$a : \iOp$} ()
				(q_1) edge node {$b : \crOp$} (q_2)
				(q_2) edge[loop above] node {$a,b : \eps$} ();
	\end{scope}
\end{tikzpicture}
	\caption{B- and S-automaton which count the maximal subsequent occurrences
of the letter $a$ in a word.}
	\label{fig:CounterAutomataExaple}
\end{figure}
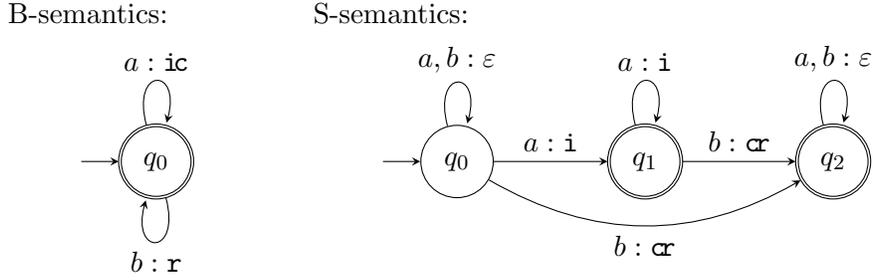

In his work \cite{regularcostfunctions}, T. Colcombet showed several properties
of the models and the functions defined by these models. First, he was able to
prove that the two automata models define, up to the equivalence 
$\costEquiv[]$, the same cost functions and concluded that boundedness 
properties are decidable. Second, he showed extensive closure properties for 
regular cost functions under standard operations such as $\min$, $\max$ or 
special kinds of projections. 

The results in the framework of regular cost functions form the basis for the
applications on the verification of infinite state systems with resources
investigated here. Therefore, we briefly repeat the major results.

\begin{thm}[Equal expressiveness, see
\cite{regularcostfunctions}]\label{thm:EqualExpressivenessOfSandB}
	Let $\automatonA$ be some B-automaton or S-automaton. The following four
    automata are effectively computable and define a function which is
    equivalent ($\costEquiv[]$) to the one defined by $\automatonA$:
	\begin{itemize}
		\item a B-automaton $\automatonA_{B}$
		\item a simple B-automaton $\automatonA_{\mathit{simB}}$
		\item an S-automaton $\automatonA_{S}$
		\item a simple S-automaton $\automatonA_{\mathit{simS}}$
	\end{itemize}
\end{thm}

Boundedness questions have been studied for cost automata since their first
introduction. These questions exist in slightly different formulations. In
our application we reduce the boundedness question for \RPRS{} over several
steps to boundedness questions for regular cost functions.

\begin{thm}[Boundedness, see
\cite{regularcostfunctions}]\label{thm:BoundednessOfCounterautomataIsDecidable}
 Let $\automatonA$ be a B- or S-automaton. The following problem is decidable:
 $$\exists M \in \nat\; \forall w \in \Sigma^* :
\semantics{\automatonA}_{B/S}(w) < M$$
\end{thm}

The class of regular languages possesses many closure properties. It is
known that it is closed under boolean operations as well as under projection.
This result was generalized and extended to regular cost functions.
In the context of cost functions, the boolean connectives conjunction and
disjunction correspond to $\max$ and $\min$. Furthermore, two variants of
projection for cost functions have been introduced: 
\emph{inf-projection} and \emph{sup-projection}. Consider an alphabet projection
$\pi: \Lambda \to \Sigma$ and let $\bar\pi: \Lambda^* \to \Sigma^*$ be the
canonical extension of $\pi$ to words. For a cost function 
$f: \Lambda^* \to \natInf$, the $\inf$-projection 
$f\infProjection{\pi}: \Sigma^* \to \natInf$ of $f$ is given by  
\[ f\infProjection{\pi}(w) := \inf_{u \in \bar\pi^{-1}(w)} f(u) \]
The $\sup$-projection is defined analogously. 

\begin{thm}[Closure properties, see \cite{regularcostfunctions}]\
\label{thm:ClosureOfRegularCostFunctions}
	\begin{enumerate}[label=(\roman*)]
		\item Let $\automatonA$ and $\automatonB$ be B/S-automata. There are
effectively computable B-automata $\automatonC_{\min}$ and
$\automatonC_{\max}$ such that $\semantics{\automatonC_{\min}}_B \costEquiv[]
\min(\semantics{\automatonA}_{B/S},\semantics{\automatonB}_{B/S})$ and
$\semantics{\automatonC_{\max}}_B \costEquiv[]
\max(\semantics{\automatonA}_{B/S},\semantics{\automatonB}_{B/S})$.
	Moreover, the automata $\automatonC_{\min}$ and $\automatonC_{\max}$ can 
	be computed exactly (without $\costEquiv[]$) if the input automata are 
	both B-automata.
		\item Let $\automatonA$ be a B/S-automaton over the alphabet $\Lambda$
and $\pi: \Lambda \to \Sigma$ be an alphabet projection function. There are
effectively computable B-automata $\automatonB_{\inf}$ and $\automatonB_{\sup}$
such that $(\semantics{\automatonA}_{B/S})\infProjection{\pi} \costEquiv[]
\semantics{\automatonB_{\inf}}_B$ and
$(\semantics{\automatonA}_{B/S})\supProjection{\pi} \costEquiv[]
\semantics{\automatonB_{\sup}}_B$.
	\end{enumerate}
\end{thm}

\noindent We remark that $\min$, $\max$ and $\inf$- as well as
$\sup$-projection preserve the equivalence property $\costEquiv[]$ on
cost functions in the following sense: for $f \costEquiv[] f'$ and $g
\costEquiv[] g'$, we have $\max(f,g) \costEquiv[] \max(f',g')$,
$\min(f,g) \costEquiv[] \min(f',g')$ and for all projections $\pi$ we
also have $f_{\inf,\pi} \costEquiv[] f'_{\inf,\pi}$ and $f_{\sup,\pi}
\costEquiv[] f'_{\sup,\pi}$.

A first, simple consequence of the previous theorem is that we can easily
modify the values of a (concrete) cost function on a regular set of its domain. For
example, we can implement a case distinction between two cost functions.
Let $f,g$ be two regular cost functions (given in form of B-automata) 
and $L$ a regular set. 
The function 
\[ h: \Sigma^* \to \natInf, w \mapsto 
	\begin{cases}
		f(w) & \text{ if } w \in L \\
		g(w) & \text{ otherwise}
	\end{cases}
\]
can be defined using the above result by $h :=
\max(\min(g,\chi_L),\min(f,\chi_{\bar L}))$ where $\bar L$ denotes the
complement of $L$.

\begin{rem}\label{rem:ClosureUnderReversing}
	Similar to regular languages, the functions defined by B- or S-automata are
	closed under reversing the words. 
	Let $\automatonA$ be an B-/S-automaton and $w^\rev = a_n\ldots a_1$ denote
    the letter-wise reversed word of $w = a_1\ldots a_n$ for $w \in \Sigma^*$
    and $a_i \in \Sigma$. There is a correction function $\alpha$ and 
    a B-/S-automaton $\automatonB$ such that 
\[
	\semantics{\automatonA}_{B/S}(w) \costEquiv
    \semantics{\automatonB}_{B/S}(w^\rev)
\]
	If no change between B- and S- automata is made, this is even possible for
	$\alpha = \mathrm{id}_{\natInf}$. 
\end{rem}
\begin{proof}
	By Theorem~\ref{thm:EqualExpressivenessOfSandB} there is a simple
B-automaton equivalent to $\automatonA$. In simple B-automata the value of a
run is the same when reading the run forward or backward. Consequently
reversing all transitions and exchanging the sets $\In$ and $\Fin$ in this 
automaton yields the desired result automaton $\automatonB$. 
The additional statement follows from the idea that one can simulate the 
reversed runs of cost automata by reversing all transitions and additionally 
storing whether the counters have been checked and then applying the check at
the other end of the respective increment block. This preserves exact values.
\end{proof}

\section{Resource Structures and the Logic FO+RR}
\label{sec:ResourceStructuresAndLogic}

In the following section, we introduce resource structures and the logic
\FORR{} as a formal tool to represent systems with resource consumption and
specify combined properties on the behavior and resource consumption
of the system. First, we define a general framework of structures and logic.
Subsequently, we show that this framework is able to express our question
about bounded reachability. Finally, we establish a connection between a
subclass of resource structures and special forms of cost automata. This
connection allows us to effectively evaluate the logic on this subclass of resource
structures.

\subsection{Resource Structures}
\label{subsec:ResourceStructures}

A \emph{resource structure} is a relational structure whose relations
incorporate a notion of resource-cost. In contrast to standard structures, the
relations are not evaluated as a set of tuples but in the form of a function that
maps every tuple to a natural number or infinity. Intuitively, the assigned
value represents the amount of resources which is needed for this tuple to be in the
given relation. A value of infinity means that a tuple is not in the relation 
at all. 

\begin{defi}
 A relational signature $\tau = \{R_1,\ldots,R_m\}$ is a set of $n_i$-ary
relational symbols. A resource structure over some signature $\tau$ is a tuple
$\structureS = (S,R_1^\structureS,\ldots,R_m^\structureS)$ consisting
of a universe $S$ and  valuations $R_1^\structureS,\ldots,R_m^\structureS$ for
the relations in $\tau$. The valuation of each relational symbol is a function
$R_i^\structureS: S^{n_i} \to \natInf$ mapping every tuple to its resource-cost
value or infinity.
\end{defi}

Resource structures can be considered an extension of ordinary
relational structures. A standard relation can be represented in the form
of the characteristic resource function which maps tuples in relation to
0 and others to $\infty$. Conversely, we define the restriction of a
resource structure $\structureS$ to some allowed resource bound $k
\in\nat$. This restricted structure is the ordinary relational structure
given by $\structureS_{\le k} := (S,R_1^{\structureS_{\le
k}},\ldots,R_m^{\structureS_{\le k}})$ with valuations defined by
$R_i^{\structureS_{\le k}} := \{ \bar s \in S^{n_i} \mid R_i^\structureS(\bar s)
\le k\}$.

As an example, we can represent an \RPRS{} as a resource structure in the form of
the configuration graph with a quantitative reachability relation. So, for a
given \RPRS{} $\prsR$ let $\resRep{\prsR} = (\Sigma^*,\tspathTo{}^{\resRep{\prsR}})$. The
resource-cost for a pair of configurations $(a,b)$ is defined to be the
minimal cost of all possible paths from $a$ to $b$. If there is no such path,
the resource-cost is set to $\infty$, i.e., $\tspathTo{}^{\resRep{\prsR}}(a,b)
:= \inf \left\{ k \in \nat \mid a \configstepsLesserCost{k} b\right\}$.

\subsection{The Logic FO+RR}
\label{subsec:LogicFORR}

Specifications for systems with resource consumption should not only be able to
express properties on the behavior of the system but also be capable of modeling
constraints on the resources. A natural question in the context of resource
structures is how many resources are needed in order to satisfy some first-order
property. We define the logic \emph{first-order+resource relations} (for short
\FORR{}) to formalize this question. Its syntax is very similar to ordinary
first-order logic but does not contain negation.

\begin{defi}[Syntax of \FORR{}]
  Let $\tau = \{R_1,\ldots,R_m\}$ be a  relational signature. \FORR{} formulas
  over $\tau$ are:
  \begin{align*}
  \phi &::= x = y \mid x \ne y \mid R_1x_1\ldots x_{n_1} \mid \ldots \mid R_mx_1\ldots
x_{n_m} \\
   &\hphantom{::= }\;\, \phi \vee \phi \mid \phi \wedge \phi \mid \forall x\phi
\mid \exists x \phi 
  \end{align*}
  We denote the set of possible \FORR{} formulas over the signature $\tau$ by
$\FORR(\tau)$.
\end{defi}

The semantics assigns to each formula a finite number or infinity instead of a
truth value. Intuitively, this number is the amount of resources which are
necessary to satisfy the formula. More precisely, it is the minimal number $k$
such that $\structureS_{\le k}$ satisfies the formula when interpreted as normal
first-order formula. This idea also explains the lack of negation in the logic.
The intuitive formulation of the semantics implies that a higher amount of
allowed resources leads to more satisfiable formulas. However, this monotonicity
is incompatible with negation. In the following, we define the formal semantics
in a way to calculate this value directly. It is easy to verify that the above
described intuition and the formal semantics below coincide.

\begin{defi}[Semantics of \FORR{}]
  Let $\varphi \in \FORR(\tau)$, $\tau = \{ R_1, \ldots R_m\}$ a signature with
$n_i$-ary relational symbols and $\structureS =
(S,R_1^\structureS,\ldots,R_m^\structureS)$ a resource structure.

  The semantics $\semantics{\varphi}^\structureS : (\free(\varphi) \to S) \to
\natInf$ of the formula $\varphi$ is a function which takes a valuation for the
free variables of the formula and maps it to a finite number or infinity. It is
defined inductively by:
  \begin{align*}
	\semantics{x = y}^\structureS(\valfunc) &:= \begin{cases} 0 & \text{if }
\valfunc(x) = \valfunc(y) \\ \infty & \text{otherwise} \end{cases} \\
	\semantics{x \ne y}^\structureS(\valfunc) &:= \begin{cases} \infty & \text{if }
\valfunc(x) = \valfunc(y) \\ 0 & \text{otherwise} \end{cases} \\
	\semantics{R_ix_1\ldots x_{n_i}}^\structureS(\valfunc) &:=
R_i^\structureS(\valfunc(x_1),\ldots,\valfunc(x_{n_i})) \\
	\semantics{\varphi_1 \vee \varphi_2}^\structureS(\valfunc) &:=
\min(\semantics{\varphi_1}^\structureS(\valfunc),\semantics{\varphi_2}
^\structureS(\valfunc)) \\
	\semantics{\varphi_1 \wedge \varphi_2}^\structureS(\valfunc) &:=
\max(\semantics{\varphi_1}^\structureS(\valfunc),\semantics{\varphi_2}
^\structureS(\valfunc)) \\
	\semantics{\exists x \varphi}^\structureS(\valfunc) &:= \inf_{a \in S}
\semantics{\varphi}^\structureS(\valfunc[x \to a]) \\
	\semantics{\forall x \varphi}^\structureS(\valfunc) &:= \sup_{a \in S}
\semantics{\varphi}^\structureS(\valfunc[x \to a]) \\
	\text{where } \valfunc[x \to a](y) &:= \begin{cases} a & \text{if } y =
x \\ \valfunc(y) & \text{otherwise} \end{cases}
  \end{align*}
  Note that $\semantics{\varphi}^\structureS$ is a constant when
$\free(\varphi) = \emptyset$. For a formula
$\varphi$ with $\free(\varphi) = \{x_1,\ldots,x_j\}$, we 
also write $\varphi(\bar x)$ and set $\semantics{\varphi(\bar
a)}^\structureS := \semantics{\varphi}^\structureS([x_i \to a_i])$ for
$\bar a \in S^j$. In general, we use the letters $x,y,z$ to indicate
variables of the logic and $a,b,c$ to indicate elements of the structure. 
\end{defi}

The formalism of \FORR{} can be used to express the bounded reachability
problem. By formalizing the definition of bounded reachability and 
simple equivalences using the formal definition of \FORR{}, we obtain:

\begin{prop}\label{prop:BoundedReachabilityIsFORRExpressible}
  Let $\structureS = (\Sigma^*,\tspathTo,\overline{A},B)$ be an extended
resource structure representation of an \RPRS{} where $A$ and $B$ are two 
sets of configurations. Let the relation $\overline{A}$ be valuated by the 
characteristic function of the complement of the set $A$, and the relation $B$ by the 
characteristic function of $B$.
The set $B$ is boundedly reachable in the \RPRS{} from $A$ if and only if
  $(\semantics{\forall x \exists y \overline{A}x \vee (By \wedge x \tspathTo
y)}^\structureS < \infty)$.
\end{prop}

\subsection{Resource Automatic Structures}
\label{subsec:ResourceAutomaticStructures}

B. Khoussainov and A. Nerode introduced the concept of automatic structures in
\cite{automatic-structures}. Automatic structures are relational structures
which are representable by automata in two aspects. First, they must have a
representation of the universe in the form of a regular language. Second, their
relations must be representable with synchronous transducers which operate over
the language representation of the universe. In \cite{automatic-structures},
automata theoretic methods are used to prove that the FO-theory of every 
automatic structures is decidable. 

We extend this concept of automatic structures to the area of resource
structures. First, we define synchronous resource transducers as a
straight-forward extension of usual synchronous transducers (see
\cite{khoussainov-automata-theory-and-app} for a comprehensive introduction on
synchronous transducers). Based on this model, we define
resource automatic structures as resource structures whose relations are
representable by these transducers. Finally, we prove that the semantics of
\FORR{} is effectively computable over resource automatic structures.

A synchronous resource transducer can be seen as a cost automaton operating
over an alphabet $\Sigma'$ consisting of vectors of elements from some original
alphabet $\Sigma$. Additionally, the vector can contain special padding symbols
(we use $\$$). With this notation, it is possible to describe a vector of words
over $\Sigma$ in the form of a word over the vector alphabet $\Sigma'$. Since the
words may have different lengths, one pads all words to the length of the
longest word either on the left- or the right end of the word.  We formalize
these ideas for the left-aligned case in the following definitions.

\begin{defi}[Word-Relations]\label{def:PaddingAndWordVectors}
	Let $\Sigma$ be a finite alphabet. In order to differentiate between words
of length $n$ and a vector of dimension $n$, the set of vectors of dimension
$n$ with entries in $\Sigma \cup \{\$\}$ is denoted by $\alphVector{\Sigma}{n}$
and the words of length $n$ still by $\Sigma^n$. Vectors with arbitrary words
as entries are denoted by $(\Sigma^*)^n$. The symbol $\pad$ (read ``pad'') is 
used for the vector consisting only of padding symbols ($\$$). We define the 
function $\bindL$ to translate between vectors of words and words over
vector-alphabets. The abbreviation \emph{conv} was introduced by Khoussainov and
Nerode and means \emph{convolution}.
\begin{align*}
	&\wsel: \Sigma^* \times \nat \to \Sigma \cup \{\$\}, (w,i) \mapsto
\begin{cases}
	a_i & \text{if $i \le |w|$ and $w = a_1 \ldots a_{|w|}$ for $a_i \in
\Sigma$} \\
	\$ & \text{otherwise}
\end{cases} \\
	&\bindL: (\Sigma^*)^n \to \alphVectorS{\Sigma}{n}, \\ 
&(w_1,\ldots,w_n) \mapsto
\left(\begin{matrix} \wsel(w_1,1) \\ \vdots \\ \wsel(w_n,1)\end{matrix}\right)
\cdots \left(\begin{matrix} \wsel(w_1,\ell) \\ \vdots \\
\wsel(w_n,\ell)\end{matrix}\right) \text{ with } \ell := \max_{i = 1,\ldots,n}
|w_i| \\
\end{align*}
	Furthermore, we define $\strip: (\Sigma \cup \{\$\})^* \to \Sigma^*$ to be
the function which just removes all occurrences of $\$$. With this, we define
$\unbind$ by
\begin{align*}
 & \unbind: \alphVectorS{\Sigma}{n} \to (\Sigma^*)^n , \\
 & 	\left(\begin{matrix} a_{1,1} \\ \vdots \\ a_{n,1}\end{matrix}\right) 
	\cdots 
	\left(\begin{matrix} a_{1,\ell} \\ \vdots \\ a_{n,\ell}\end{matrix}\right)
	\mapsto (\strip(a_{1,1}\ldots a_{1,\ell}), \ldots, \strip(a_{n,1}\ldots
a_{n,\ell}))
\end{align*}

We call a word over the vector alphabet correctly padded if it is generated
by some tuple of words. Formally, we denote the
set of all left-aligned correctly padded words over the $n$ dimensional vector 
alphabet by $\alphVectorL{\Sigma}{n} := \bindL((\Sigma^*)^n)$. Accordingly, the complement
(the set of not correctly padded words) is called $\NalphVectorL{\Sigma}{n}$.
We remark that correctly padded words do not contain the $\pad$ symbol and 
in every component there are only $\$$ after the first position where $\$$ 
occurs. Moreover, the correctly padded words $\alphVectorL{\Sigma}{n}$ form
a regular set.

To simplify the notation in the constructions we write $\padprodL$ to
combine two correctly padded words from vector alphabets
$\alphVectorL{\Sigma}{n_1}$, $\alphVectorL{\Sigma}{n_2}$ to a new correctly
padded word from $\alphVectorL{\Sigma}{n_1+n_2}$ which combines the two vectors.
Formally, this can also be written in the form $\bar v \padprodL \bar w =
\bindL((\unbind(\bar v),\unbind(\bar w)))$.
\end{defi}

\begin{defi}[Synchronous Resource
Transducer]\label{def:SynchronousResourceTransducer}
  A synchronous resource transducer for an $n$-ary relation $R \subseteq
(\Sigma^*)^n$ over a finite alphabet $\Sigma$ is a B-automaton $\automatonT$
operating over the alphabet $\alphVector{\Sigma}{n}$.

  The semantics is given by:
  \begin{align*}
   \semantics{\automatonT}_{\padprodL} &: (\Sigma^*)^n \to \natInf, \bar w
      \mapsto \semantics{\automatonT}_B(\bindL(\bar w)) 
  \end{align*}

  We also define $\semantics{\automatonT}_{\padprodL}^S$ which is identical
to the above definition but uses S-automaton semantics instead of B-automaton
semantics on the automaton $\automatonT$. We remark that in addition to this
relational semantics, we sometimes view these automata as plain cost automata
over $\alphVector{\Sigma}{n}$. This point of view is mainly used in
constructions.
\end{defi}

We remark that we could repeat all the definitions with index $\mathsf{R}$ for
right-aligned relations. These relations are essentially the (word-wise) reversed 
relations of left-aligned ones. For a quantiative relation $R: \Sigma^* \times
\Sigma^* \to \natInf$, let $R^\rev(u,v) := R(u^\rev,v^\rev)$. By 
Remark~\ref{rem:ClosureUnderReversing}, it is easy to see that $R$ can be 
defined by a left-aligned B-automaton iff $R^\rev$ can be defined by a 
right-aligned B-automaton. 
We follow the notation in most of the literature
regarding automatic structures and present the definitions and ideas around 
resource automatic structures with left-aligned relations. However, we show 
that there is no conceptual difference to right-alignment and use right-aligned 
relations in the context of prefix replacement systems because we consider 
it to be more natural.

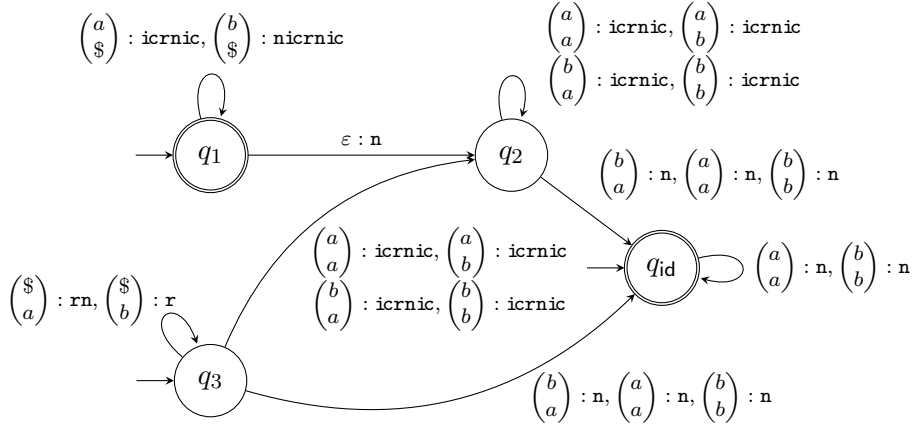
\begin{figure}[t]
	\begin{center}
			\begin{tikzpicture}[initial text=,>=stealth]
	\tikzstyle{ln}=[auto,font={\scriptsize}]

	\node[state,initial,accepting] (q1) at (2,1.5) {$q_1$};
	\node[state] (q2) at (6,1.5) {$q_2$};
	\node[state,initial] (q3) at (2,-1.5) {$q_3$};
	\node[state,initial,accepting] (qId) at (8,0) {$q_{\mathsf{id}}$};
	
	\draw[->] 	(q1) edge[loop above] node[ln] {$\twovec{a}{\$} : \icOp\rOp\nOp\icOp, \twovec{b}{\$} : \nOp\icOp\rOp\nOp\icOp$} (q1)
				(q1) edge node[ln] {$\eps:\nOp$} (q2)
				(q2) edge[loop above] node[ln,right=3mm,yshift=3mm] {$\begin{matrix}	\twovec{a}{a} : \icOp\rOp\nOp\icOp, \twovec{a}{b} : \icOp\rOp\nOp\icOp \\ 
																		\twovec{b}{a} : \icOp\rOp\nOp\icOp, \twovec{b}{b} : \icOp\rOp\nOp\icOp \end{matrix}$} ()
				(q2) edge node[ln] {$\twovec{b}{a} : \nOp, \twovec{a}{a} : \nOp,\twovec{b}{b} : \nOp$} (qId);

	\draw[->]	(q3) edge[bend right] node[ln,pos=0.7,below right=-2mm] {$\twovec{b}{a} : \nOp, \twovec{a}{a} : \nOp,\twovec{b}{b} : \nOp$} (qId)
				(q3) edge[in=110,out=140,loop] node[ln,above left=-4mm] {$\twovec{\$}{a} : \rOp\nOp, \twovec{\$}{b} : \rOp$} ()
				(q3) edge[bend left] node[ln,below right=-2mm,pos=0.45] {$\begin{matrix}	\twovec{a}{a} : \icOp\rOp\nOp\icOp, \twovec{a}{b} : \icOp\rOp\nOp\icOp \\ 
																				\twovec{b}{a} : \icOp\rOp\nOp\icOp, \twovec{b}{b} : \icOp\rOp\nOp\icOp \end{matrix}$} (q2);
	
	\draw[->] (qId) edge[loop right] node[ln] {$\twovec{a}{a} : \nOp, \twovec{b}{b} : \nOp$} ();
\end{tikzpicture}
	\end{center}
	\caption{Synchronous resource transducer for the reachability-cost of the
\RPRS{} in Figure~\ref{fig:RPRSExample}}
	\label{fig:ExampleSyncResTransducer}
\end{figure}

To illustrate and motivate the previous definition, we give an example
transducer in Figure~\ref{fig:ExampleSyncResTransducer}. In order to simplify
the presentation of the automaton, we use the usual concept of $\eps$-transitions 
although they are not originally part of the definition 
of B-automata. An $\eps$-transition changes the state and executes the associated counter
operation but does not consume symbols from the input word. It is known that 
$\eps$-transitions do not change the expressive power of the model and that it
is possible to algorithmically obtain an equivalent B-automaton. The idea behind 
this $\eps$-elimination procedure is similar to the one for normal $\eps$-NFAs. The 
transducer in Figure~\ref{fig:ExampleSyncResTransducer} computes the 
reachability-cost relation of the example \RPRS{} 
whose configuration graph is shown in Figure~\ref{fig:RPRSExample}. It 
operates on right-aligned words since this is more common for prefix
replacement systems. This behavior resembles the replacement 
operation which allows to replace some prefix followed by a common postfix. 
We remarked earlier, that it is possible to remove an arbitrary number of 
$a$s or $b$s in front of a word with a resource-cost of at most 2. This is 
reflected in the loops of $q_1$ which are labeled with the counter operations 
$(\nOp)\icOp\rOp\nOp\icOp$. This sequence of counter operations corresponds 
to one removal step with detour. We remark that we left the $\nOp$s in the counter
sequence although they are not part of B-automaton counter operations. 
This way, the counter operations in the automaton  
resemble more closely the accumulated operations of all replacement steps. In a similar 
way we can also replace one letter by another. However, this requires first to 
pop the stack until the letter which should be changed and afterwards to push the old 
contents again. We saw that only the pop operations influence the 
resource-cost. This is handled in the state $q_2$ which calculates the cost of
the pop operations, the push operations match the number of pop operations but
do not cost anything and thus do not occur explicitly in the transitions.
The state $q_3$ covers the ``free'' addition at the 
beginning of a word. Nevertheless, in all cases we have to ensure that the 
last $b$ in the word is not replaced by an $a$ since this is not possible in 
the \RPRS{}. This is ensured by the transitions to the state 
$q_{\mathrm{id}}$. This state finally recognizes the identity with no more 
resource-cost after the replaced front was completely read. Although this 
particular transducer was constructed manually, we see in 
Section~\ref{sec:ReachabilityWithAnnotations} how to derive such a 
transducer algorithmically. We remark that the algorithmic construction
works in a completely different way. 

We define \emph{resource automatic structures} as a combination of the
ideas of resource- and automatic structures. Similar to automatic structures,
we require a representation of the universe in the form of a regular language.
Additionally, the resource relations are required to have a synchronous
resource transducer which computes their semantics. 

\begin{defi}[Resource Automatic Structure]
A resource structure $\structureS = (S,R_1^\structureS,\ldots,R_m^\structureS)$
is called resource automatic if it satisfies two conditions. First, there is a
finite alphabet $\Sigma$ such that $S \subseteq \Sigma^*$ is a regular language.
Second, for all relations $R_i^\structureS$ there exists a synchronous resource
transducer $\automatonT_{R_i}$ such that $R_i^\structureS(\bar a) =
\semantics{\automatonT_{R_i}}_{\padprodL}(\bar a)$.

Additionally, we also call a structure resource automatic if it is isomorphic to
a resource automatic structure as defined above.
This way, we also allow for resource automatic structures with a universe 
which is not a word-language.
\end{defi}

We first remark that for every resource automatic structure $\structureS$ and
every $k \in \nat$ the structure $\structureS_{\le k}$ is automatic. As
remarked earlier, it is possible to simulate the behavior of cost automata
up to a fixed bound $k$ with normal finite automata by extending their state
space. Thus, one can easily obtain synchronous transducers for the relations in
$\structureS_{\le k}$ from the resource transducers.

\begin{rem}\label{rem:ResourceAutomaticRightAligned}
	One can also consider resource automatic structures given by right-aligned
	synchronized resource transducers. These structures are equally expressive
	as resource automatic structures given by left-aligned transducers.
\end{rem}
\begin{proof}
	As remarked earlier, if there is a right-aligned (B-)resource transducer for 
	some relation $R$, there is a left-aligned resource transducer for 
	$R^\rev$. Consequently, one directly obtains an isomorphic structure
	given by left-aligned resource transducers (with the isomorphism $\vartheta:
	S \to S^\rev, w \mapsto w^\rev$). 
\end{proof} 

The rest of this section is dedicated to show how we can compute the semantics
of \FORR{} formulas in a given resource automatic structure. The proof is
divided into a lemma followed by the main result. The proof of the main
result is quite similar to the decidability proof of \FO{} over automatic
structures. We inductively construct synchronous resource transducers which
(approximately) compute the semantics of \FORR{} formulas with free variables.
In Section~\ref{sec:Preliminaries}, we saw that the $\min$ and $\max$ of two
cost automata is again representable by a cost automaton. This directly
enables a simulation of $\wedge$ and $\vee$ in \FORR{} formulas on the level of
automata. Although we also saw that $\sup$ and $\inf$-projection can be realized
by cost automata, these operations do not directly correspond to the
semantics of $\forall$ and $\exists$ in \FORR{}. The reason for this is that
alphabet projection does not preserve the encoding of convoluted words. If one
projects away the longest word in a convolution of words, the result contains
a sequence of $\pad$ symbols at the end. Hence, the result is not correctly
padded. In the following lemma we show how to solve this problem.
The main idea
is to divide the calculation of the semantics of $\forall$ and $\exists$ into
two $\sup$ or $\inf$ operations by the following idea. We use
$\sup$- and $\inf$-projection in a first step and add an additional step to
cover sequences of $\pad$ symbols at the end. In this second step, we use that
computing the $\inf$ and $\sup$ over all runs is part of B- and S-automaton
semantics. Thus, we can implement such a computation by adding new transitions
to an automaton.

\begin{lem}\label{lem:InfAndSupLikeFORRAreComutable}
	Let $\Sigma$ be a finite alphabet, $\automatonT$ a synchronous resource
transducer operating over $\alphVectorS{\Sigma}{n+1}$. There are effectively 
computable synchronous resource transducers $\automatonT_{\sup}$, 
$\automatonT_{\inf}$ operating over $\alphVectorS{\Sigma}{n}$ and a correction
function $\alpha$ such that for all $\bar u \in \alphVectorL{\Sigma}{n}$:
\begin{enumerate}[label=(\roman*)]
	 \item $\semantics{\automatonT_{\sup}}_B(\bar u) \costEquiv
   \sup\limits_{w \in \Sigma^*} \semantics{\automatonT}_B(\bar u \padprodL w)$
	\item $\semantics{\automatonT_{\inf}}_B(\bar u) \costEquiv \inf\limits_{w
     \in \Sigma^*} \semantics{\automatonT}_B(\bar u \padprodL w)$
\end{enumerate}
\end{lem}
\begin{proof}
The proof of both parts of the lemma is divided into two similar parts. First,
we show how to transform the single $\sup$ or $\inf$ operation into two
operations of which the first is the $\sup$- or $\inf$-projection as defined in
\cite{regularcostfunctions}. In the second part, we exploit the definition of
cost automata to compute the remaining ``$\sup_{\bar p \in \pad^*}$'' and
``$\inf_{\bar p \in \pad^*}$''.

We start with part (i) of the statement in the lemma. So, let $f :
\alphVectorS{\Sigma}{n+1} \to \natInf$ be the function defined by $\automatonT$
when interpreted as normal B-automaton. Since $\NalphVectorL{\Sigma}{n+1}$ is a
regular set, we can assume w.l.o.g. that $f(\bar u) = 0$ for all $\bar u \in
\NalphVectorL{\Sigma}{n+1}$. This ensures that the value of the $\sup$-projection
is determined only by correctly encoded (vector-)words. Let furthermore $\pi:
\alphVector{\Sigma}{n+1}\to \alphVector{\Sigma}{n}$ be the projection function
removing the last component. With the above explained argument, we can divide
the $\sup$ in the following way. We obtain that for
all $\bar u \in \alphVectorL{\Sigma}{n}$:
\[
  \sup_{w \in \Sigma^*} \semantics{\automatonT}_B(\bar u \padprodL w) = 
  \sup_{w \in \Sigma^*} f(\bar u \padprodL w) =  \sup_{\bar p \in \pad^*}
    \underbrace{\sup_{\twosm{\bar v \in
\alphVectorS{\Sigma}{n+1}:}{\bar\pi(\bar v)
     = \bar u \bar p}} f(\bar v)}_{\sup\text{-projection}}
\]
 
By Theorem~\ref{thm:ClosureOfRegularCostFunctions}, we obtain a
B-automaton $\automatonT'$ and by Theorem~\ref{thm:EqualExpressivenessOfSandB}
an S-automaton $\automatonT'_S$ operating over $\alphVectorS{\Sigma}{n}$ such
that for some correction functions $\beta$ and $\beta'$:
\[
 \sup_{w \in \Sigma^*} \semantics{\automatonT}_B(\bar u \padprodL w)
 \costEquiv[\beta] 
   \sup_{\bar p \in \pad^*} \semantics{\automatonT'}_B(\bar u  \bar p)
 \costEquiv[\beta'] 
  \sup_{\bar p \in \pad^*} \semantics{\automatonT'_S}_S(\bar u \bar p)
\]
We now show how to construct an S-automaton $\automatonT'_{\sup}$ which computes
the last $\sup$. The idea of this construction exploits the special semantics
of S-automata. For S-automata, the value of a word is the $\sup$ of the values
of all accepting runs. Thus, we can construct an automaton computing $\sup_{\bar
p \in \pad^*} \semantics{\automatonT'_S}_S(\bar u \bar p)$ by making all runs of
$\automatonT'_S$ on some $\bar u \bar p$ also accepting for $\bar u$. We
implement this by introducing $\eps$-transitions in positions with
$\pad$-transitions.  Formally, let $\automatonT'_S =
(Q,\alphVector{\Sigma}{n},\Delta,\In,\Fin,\Gamma)$. We define
$\automatonT_{\sup}' = (Q \times \{0,1\}, \alphVector{\Sigma}{n}, \Delta', 
\In \times \{0\}, \Fin \times \{1\},\Gamma)$ with:
	\begin{align*}
	  \Delta' &:= \{ ((p,0),a,(q,0),\fraku) \mid (p,a,q,\fraku) \in \Delta \} \\
			  &\;\cup \{ ((p,1),\eps,(q,1),\fraku) \mid (p,\pad,q,\fraku) \in
              \Delta\} \\
			  &\;\cup \{ ((p,0),\eps,(p,1),\fraku) \mid p \in Q, \fraku(\gamma) := \eps \text{ for all } \gamma \in \Gamma\}
	\end{align*}
	Let $\rho$ be an accepting run of $\automatonT'_S$ on $\bar u \pad^k$.
We construct the run $\rho'$ of $\automatonT_{\sup}'$ on $\bar u$ out of the run
$\rho$. First, copy the $\bar u$ part of the run $\rho$ into the first
component of the state vector and set the second component to $0$. Subsequently,
take the $\eps$-transition to change the second component to $1$. Then, copy the
$\pad^k$ part of the run $\rho$ into the first component and set the second
component to $1$. This is possible because of the $\eps$-transitions which are
inserted at the positions of the $\pad$-transitions. By the construction of the
transition relation, this is a valid accepting run of $\automatonT_{\sup}'$ on
$\bar u$ which induces the same counter sequence as $\rho$ and thus has the
same cost-value. Consequently, $\sup\limits_{\bar p \in \pad^*}
\semantics{\automatonT'_S}_S(\bar u \bar p) \le
\semantics{\automatonT_{\sup}'}_S(\bar u)$.

	Conversely, let $\rho$ be an accepting run of $\automatonT_{\sup}'$ on $\bar
u$. Since all final states have a $1$ in their second component and the run 
can only change the second component from $0$ to $1$ (and not back), the run
$\rho$ can be split into a part which uses states with a $0$ in the second
component and part which uses states with a $1$ there. By construction, the
first part consumes $\bar u$ and the second part uses $\eps$-transitions at
positions with $\pad$-transitions in $\automatonT'_S$. Consequently, it is
possible to construct a run $\rho'$ of $\automatonT'_S$ on $\bar u \pad^k$ for
some $k \in \nat$ which induces the same counter sequence as $\rho$. Therefore,
$\sup\limits_{\bar p \in \pad^*} \semantics{\automatonT'_S}_S(\bar u \bar p) \ge
\semantics{\automatonT_{\sup}'}_S(\bar u)$.

We use a rather lengthy procedure to eliminate $\eps$-transitions in S-automata,
which is described in more detail in \cite{la11}. The main problem in 
$\eps$-elimination for S-automata arises from the fact that loops of 
$\eps$-transitions with $\iOp$ counter operations are meaningful for the 
semantics of S-automata since the supremum over all runs is built. Consequently,
one could loop more and more often through the $\eps$-increment-loop in order
to obtain large counter values before the next check. The $\eps$-elimination 
procedure systematically searches for these loops and adds a control component 
to the automaton indicating that a certain counter may have an arbitrarily 
large value (due to looping in $\eps$-increment-loops). 
If a counter that is marked to have such an arbitrary large value hits a 
$\crOp$ operation, it just gets reset but not checked. 

First note that for every run $\rho$ of the $\eps$-automaton, there is a similar 
run $\rho'$ of the $\eps$-free automaton with at least the same value because
skipping a check can only increase the value of the run. 

For the converse, consider a run $\rho$ of the automaton with eliminated 
$\eps$-transitions that skips a check because the counter was indicated to be 
arbitrarily large. We distinguish two cases. First, if no counter is checked,
the value of $\rho$ is $\infty$. However, if we construct
a run $\rho'$ similar to $\rho$ on the $\eps$-automaton, the counter is checked
at the position where the check is skipped in $\rho$. In order to obtain the
value $\infty$ on the $\eps$-automaton, we construct a sequence $\rho_k'$ of
runs such that each run loops $k$ 
times through the $\eps$-increment-loop before the check. All these runs are valid
for the input word and yield a value of at least $k$. Since the value of a word on an
S-automaton is determined by the supremum over all runs, the value is $\infty$
as the value of $\rho$ in this case. Second, if there is some counter checked
in $\rho$, the value of $\rho$ is some finite $k$. When transferring the run
$\rho$ to a similar run $\rho'$ of the $\eps$-automaton, all the counters 
checked in $\rho$ are also checked in $\rho'$ with the same values. 
Consequently, we just have to ensure that the
additional check, which is skipped in $\rho$ but occurs in $\rho'$,
does not decrease the value of the run. To achieve this, we construct the run
$\rho'$ such that it loops $k$ times through the $\eps$-increment-loop before 
the check. Thereby, the counter has at least value $k$ before
the additional check occurs and the value of $\rho'$ stays $k$ because the
value of a run is determined by the minimal checked counter value in S-automata.

By Theorem~\ref{thm:EqualExpressivenessOfSandB}, we obtain an equivalent
B-automaton $\automatonT_{\sup}$ and thus obtain for appropriate correction 
functions $\delta$ and $\beta''$ in total:
\[
 \semantics{\automatonT_{\sup}}_B(\bar u) \costEquiv[\delta]
 \semantics{\automatonT_{\sup}'}_S(\bar u) = \sup\limits_{\bar p \in \pad^*}
 \semantics{\automatonT'_S}_S(\bar u \bar p) \costEquiv[\beta''] \sup_{w \in
 \Sigma^*} \semantics{\automatonT}_B(\bar u \padprodL w) 
\]
 
We now prove part (ii) of the lemma with similar techniques. However, we now
exploit the semantics of B-automata to implement the additional
$\inf$-computation.  So, let $f$ and $\pi$ be like in the previous case but now
assume w.l.o.g. that all words $\bar u \in \NalphVectorL{\Sigma}{n+1}$ have
$f(\bar u) = \infty$. For the same reasons as above, we obtain
\[ 
  \inf_{w \in \Sigma^*} \semantics{\automatonT}_B(\bar u \padprodL w) =
  \inf_{w \in \Sigma^*} f(\bar u \padprodL w) =  \inf_{\bar p \in \pad^*}
   \underbrace{\inf_{\twosm{\bar v \in \alphVectorS{\Sigma}{n+1}:}{\bar\pi(\bar v) =
   \bar u \bar p}} f(\bar v)}_{\inf\text{-projection}}
\]

Again by Theorem~\ref{thm:ClosureOfRegularCostFunctions} and
Theorem~\ref{thm:EqualExpressivenessOfSandB}, we obtain a simple B-automaton
$\automatonT'$ which computes the $\inf$-projection:
\[ 
 \inf_{w \in \Sigma^*} \semantics{\automatonT}_B(\bar u \padprodL w)
  \costEquiv[\beta] \inf_{\bar p \in \pad^*} \semantics{\automatonT'}_B 
 (\bar u \bar p) 
\]
We now create the automaton $\automatonT_{\inf}$ by taking $\automatonT'$ and
making all states final from which one can reach a final state with only
$\pad$-transitions. This yields an automaton which computes a function which is
slightly different from the real $\inf$ but still correct w.r.t. a
correction function $\delta$. Let $m$ be the number of states in
$\automatonT'$ and $\delta(x) = x + m$. We show that
\[
     \semantics{\automatonT_{\inf}}_B(\bar u) \costlea[\delta] 
     \inf_{\bar p \in \pad^*} \semantics{\automatonT'}_B(\bar u \bar p) 
     \text{ and }
     \inf_{\bar p \in \pad^*} \semantics{\automatonT'}_B(\bar u \bar p)
     \costlea[\delta]
     \semantics{\automatonT_{\inf}}_B(\bar u) 
\]
For the first inequality note that $\natInf$ is well-ordered. Thus, there
is a $k \in \nat$ such that $\semantics{\automatonT'}_B(\bar u \pad^k) = c$
assumes the value of the infimum. W.l.o.g. this infimum is smaller than $\infty$
(otherwise there is nothing to show in the first inequality). By the definition
of the model, there is an accepting run $\rho$ of $\automatonT'$ on $\bar
u\pad^k$ such that the cost-value of this run is $c$. Let $\rho'$ the front part
of $\rho$ which reads $\bar u$, $q$ the state after reading $\bar u$ and $c_q$
the maximal checked counter value at this point in the run. By the definition of
the semantics of B-automata, the maximal checked counter value can only increase
in the course of a run. Moreover, the run $\rho$ shows that the state $q$ can
reach a final state with only $\pad$-transitions. Therefore, $\rho'$ is an
accepting run of $\automatonT_{\inf}$ with the value $c_q$ and
\[
   \semantics{\automatonT_{\inf}}_B(\bar u) \le c_q \le c =
   \semantics{\automatonT'}_B(\bar u \pad^k) = 
   \inf\limits_{\bar p \in \pad^*} \semantics{\automatonT'}_B(\bar u \bar p) \le
   \delta\left(\inf\limits_{\bar p \in \pad^*} \semantics{\automatonT'}_B
   (\bar u \bar p)\right)
\]
Conversely, let now $\rho$ be an accepting run of $\automatonT_{\inf}$ on $\bar
u$ with cost-value $c = \semantics{\automatonT_{\inf}}_B(\bar u)$. If there is
no such run, we have $c = \infty$ and there is nothing to show. By construction,
there is a final state $q_f$ of $\automatonT'$ which is
reachable from $q$ with $k \ge 0$ $\pad$-transitions. Since there is a loop-free
path, we have $k \le m$. We look at the run $\rho'$ which is created by
appending these $\pad$-transitions to $\rho$. The run $\rho'$ is accepting for
$\automatonT'$ on $\bar u \pad^k$. Since every transition has at most one
$\icOp$-operation ($\automatonT'$ is simple), the cost-value of $\rho'$ is
limited by $c + k = \semantics{\automatonT_{\inf}}_B + k$. Altogether:
\[
 \inf\limits_{\bar p \in \pad^*} \semantics{\automatonT'}_B(\bar u \bar p) 
 \le  \semantics{\automatonT'}_B(\bar u \pad^k) \le c + k 
 \le \semantics{\automatonT_{\inf}}_B(\bar u) + m 
 = \delta\left(\semantics{\automatonT_{\inf}}_B(\bar u)\right)
\]
Alternatively, it also would have been possible to use the same approach as for
part (i). However, the presented way has the advantage that a costly
$\eps$-elimination procedure is not necessary. The major reason why such an
approach is not possible for the first part of the lemma is the semantics of
loops consisting only of $\eps$-transitions in cost automata. In S-automata,
paths with high values are preferred. Consequently, it would change the value
of a run to loop through some increment again and again. In the case of
B-automata it is sufficient to take it once or twice in order to obtain a low
counter value with a reset located on the loop. 
\end{proof}

We now have all prerequisites for a concise formulation of the main theorem on
resource automatic structures. We first inductively translate a formula
$\varphi$ with free variables into a synchronous resource transducer which
calculates a function that is $\alpha$-equivalent to the function defined by the
semantics of $\varphi$. Finally, we
explain why the equivalence relation $\costEquiv[]$ is no real restriction and
how to calculate precise values of the semantics.

\begin{thm}\label{thm:FORRisEffectivelyComputableOnRAStructures}
 There is an algorithm which takes as input a resource automatic structure 
$\structureS = (\Sigma^*,R_1^\structureS,\ldots,R_m^\structureS)$ in form of resource
transducers and an \FORR{} formula with at least one free variable over 
the signature of $\structureS$ and outputs a synchronous resource transducer 
that defines an equivalent function to $\semantics{\varphi}^\structureS$. 
\end{thm}
\begin{proof}
Let $\varphi(\bar x)$ be the formula with $n > 0$ free variables.

We show by induction on the structure of the formula how to construct a
synchronous resource transducer $\automatonT_{\varphi}$ such that for some
correction function $\alpha$:
\[ 
  \forall \bar a \in (\Sigma^*)^n: \semantics{\varphi(\bar a)}^\structureS \costEquiv
  \semantics{\automatonT_{\varphi}}_{\padprodL}(\bar a) 
\]

\inductionstart{Let $\varphi = (x = y)$}
  The following transducer obviously captures the semantics of $x = y$:
\[
 \automatonT_{x = y} = (\{q\},\alphVector{\Sigma}{2},\{q\},\{q\},\{\gamma_0\},
 \{(q,(a,a),q,\fraku) \mid a \in \Sigma, \fraku(\gamma_0) := \eps\})
\]
\inductionend\inductionstart{Let $\varphi = (x \ne y)$}
  The following transducer captures the semantics of $x \ne y$:
\[
 \automatonT_{x \ne y} = (\{q_=,q_{\ne},q_l,q_r\},\alphVector{\Sigma}{2},\{q_=\},\{q_{\ne},q_l,q_r\},\{\gamma_0\},
 \Delta)
\]
where $\Delta$ is given as follows with $\fraku$ s.t. $\fraku(\gamma_0) = \eps$
\begin{align*}
	\Delta &= \{ (q_=,(a,a),q_=,\fraku),(q_{\ne},(a,a),q_{\ne},\fraku)  \mid a \in \Sigma \} \\
			&\cup \{ (q_=,(a,b),q_{\ne},\fraku) \mid a,b \in \Sigma, a \ne b \} \\
			&\cup \{ (q_=,(a,\$),q_l,\fraku),(q_l,(a,\$),q_l,\fraku)  \mid a \in \Sigma \} \\
			&\cup \{ (q_=,(\$,a),q_r,\fraku),(q_r,(\$,a),q_r,\fraku)  \mid a \in \Sigma \}
\end{align*}

\inductionend\inductionstart{Let $\varphi = R_ix_1\ldots x_{n_i}$}
  By the definition of a resource automatic structure, there is a transducer
  $\automatonT_{R_i}$ such that  
\[
    \semantics{\automatonT_{R_i}}_{\padprodL}(\bar a) =
    R_i^\structureS(\bar a) = 
    \semantics{R_ia_1\ldots a_{n_i}}^\structureS =
    \semantics{\varphi(\bar a)}^\structureS 
\]

\inductionstep{Let $\varphi = \varphi_1 \vee \varphi_2$}
By the induction hypothesis, there are transducers $\automatonT_{\varphi_1}$,
$\automatonT_{\varphi_2}$ such that $\semantics{\varphi_1(\bar b)}^\structureS
\costEquiv[\beta] \semantics{\automatonT_{\varphi_1}}_{\padprodL}
(\bar b)$ and $\semantics{\varphi_2(\bar c)}^\structureS
\costEquiv[\delta] \semantics{\automatonT_{\varphi_2}}_{\padprodL}(\bar c)$. The
free variables of $\varphi$ consist of the free variables of $\varphi_1$
combined with the free variables of $\varphi_2$. Let $\bar x =
\{x_1,\ldots,x_n\}$ be the free variables of $\varphi$. In a
first step, the two automata $\automatonT_{\varphi_1}$ and
$\automatonT_{\varphi_2}$ have to be adapted such that they take all free
variables of $\varphi$ as input. This can easily be achieved by adding new
components to the input vector which are not taken into account by the
automaton. Furthermore, we possibly have to reorder the components in the input
of the automaton. This can be achieved by reordering them in the transition
relation accordingly. As a result, we obtain automata $\automatonT_{\varphi_1}'$
and $\automatonT_{\varphi_2}'$ operating over the input alphabet
$\alphVector{\Sigma}{n}$ such that:
\[
  \semantics{\automatonT_{\varphi_1}'}_{\padprodL}(\bar a)
  \costEquiv[\beta] \semantics{\varphi_1}^\structureS([x_i \to a_i])
  \text{ and }
  \semantics{\automatonT_{\varphi_2}'}_{\padprodL}(\bar a)
  \costEquiv[\delta] \semantics{\varphi_2}^\structureS([x_i \to a_i])
\]
By Theorem~\ref{thm:ClosureOfRegularCostFunctions}, we can construct an
automaton $\automatonT_{\varphi}$ such that the equation
$\semantics{\automatonT_{\varphi}}_B =
\min(\semantics{\automatonT_{\varphi_1}'}_B,\semantics{\automatonT_{\varphi_2}'}
_B)$ holds. For this automaton, we have
\begin{align*}
 \semantics{\automatonT_{\varphi}}_{\padprodL}(\bar a) 
   &= \min\left(
     \semantics{\automatonT_{\varphi_1}'}_{\padprodL}(\bar a),
     \semantics{\automatonT_{ \varphi_2 }'} (\bar a)
    \right) \\
	&\costEquiv 
	\min\left( 
	 \semantics{\varphi_1}^\structureS([x_i \to a_i]),
     \semantics{\varphi_2}^\structureS([x_i \to a_i]) 
    \right) \\
    &= \semantics{\varphi(\bar a)}^\structureS
\end{align*}

\inductionstep{Let $\varphi = \varphi_1 \wedge \varphi_2$}
	This step is analog to the previous one when replacing $\min$ by $\max$. 
\inductionstep{Let $\varphi = \forall y \psi$}
	By the induction hypothesis, there is an automaton $\automatonT_{\psi}$
such that $\semantics{\automatonT_{\psi}}_{\padprodL}(\bar c) \costEquiv[\beta]
\semantics{\psi(\bar c)}^\structureS$. We assume w.l.o.g that the free variable
$y$ is represented by the last component of the input vector of
$\automatonT_{\psi}$. By Lemma~\ref{lem:InfAndSupLikeFORRAreComutable}, there
is an automaton $\automatonT_{\varphi}$ such that
$\semantics{\automatonT_{\varphi}}_B(\bar u) \costEquiv[\delta] \sup\limits_{w
\in \Sigma^*} \semantics{\automatonT_{\psi}}_B(\bar u \padprodL w)$ for all
$\bar u \in \alphVectorL{\Sigma}{n}$. Altogether, we have:
\begin{align*}
	\semantics{\automatonT_{\varphi}}_{\padprodL}(\bar a) &=
    \semantics{\automatonT_{\varphi}}_B(\bindL(\bar a)) \\
	&\costEquiv[\delta] \sup\limits_{w \in \Sigma^*}
   \semantics{\automatonT_{\psi}}_B(\bindL(\bar a) \padprodL w) \\
    &= \sup\limits_{b \in S} \;
    \semantics{\automatonT_{\psi}}_{\padprodL}(\bar a, b) \\
    &\costEquiv[\beta] \sup_{b \in S} \semantics{\psi(\bar a, b)}^\structureS \\
    &= \semantics{\forall y \psi(\bar a)}^\structureS =
    \semantics{\varphi(\bar a)}^\structureS
\end{align*}

\inductionstep{Let $\varphi = \exists y \psi$}
	This step is again analog to the previous one. We just use the second part
of Lemma~\ref{lem:InfAndSupLikeFORRAreComutable} instead of the first part.\qedhere
\inductionend
\end{proof}

\noindent It remains to explain how to cover \FORR{} sentences. We remind the reader of 
the fact that the value of an \FORR{} sentence is a constant. 
Calculating this constant value of a sentence up to the equivalence relation 
$\costEquiv[]$ means just checking whether the value is infinite or not.
We present this separately from the case with free variables in order to 
emphasize the computational steps necessary to decide whether a formula 
has a finite value.  

\begin{cor}
	Let $\structureS$ be a resource automatic structure like in the previous 
	theorem and $\varphi$ an \FORR{} sentence over the signature of $\structureS$.
	It is decidable whether $\semantics{\varphi}^\structureS < \infty$. 
\end{cor}
\begin{proof}
First, we remove the outermost quantifier from the sentence. Now we can
apply Theorem~\ref{thm:FORRisEffectivelyComputableOnRAStructures} and get 
a synchronized resource transducer $\automatonT$ operating over $\Sigma$. 
If the outermost quantifier is existential, the formula has a finite value 
if and only if there is some accepting run on $\automatonT$. 
If the outermost quantifier is universal, the automaton is accepting if and 
only if the automaton is bounded. This can be checked by 
Theorem~\ref{thm:BoundednessOfCounterautomataIsDecidable}.
\end{proof}

The decision procedure described above paves the way for computing precise
values of \FORR{} formulas. First, one can check whether the value of a
formula $\varphi$ (for some possible valuation) is infinite by the presented
decision procedure. If it is infinite, we are done because $\costEquiv[]$
preserves boundedness and thus the precise value of this formula is also
$\infty$. Otherwise, it is known that the formula is bounded and there is some
$k \in \nat$ with $\semantics{\varphi}^\structureS = k$. We remarked earlier
that this $k$ is exactly the smallest $k$ such that $\structureS_{\le k} \models
\varphi$ when $\varphi$ is interpreted as normal \FO{}-formula. Since the
structures $\structureS_{\le k}$ are (standard) automatic structures for all
fixed $k$, this can be checked algorithmically. Thus, one can check the above
conditions for all possible $k$ and it will terminate because we already know
that the value is finite.

We remark that the previous theorem only covered resource automatic
structures whose universe contains all words in $\Sigma^*$. However, this is no real
restriction. By definition, the universe $S \subseteq \Sigma^*$ is a regular
language. Consequently, we can extend the universe to full $\Sigma^*$ and
set the value of all relations for elements outside of $S$ to $\infty$. 
Moreover, we introduce two new relations $S$ and $\overline{S}$ valuated
with the characteristic functions of $S$ and its complement language. 
Similar to standard first-order logic, it is then possible to adapt $\FORR$
formulas by restricting the quantifications to elements of $S$ by 
transforming $\exists x \varphi$ into $\exists x (Sx \wedge \varphi)$ 
and $\forall x \varphi$ into $\forall x (\overline{S}x \vee \varphi)$. 
With the standard arguments, one can show that this provides a reduction
of arbitrary resource automatic structures to those with universe $\Sigma^*$.

\section{Computing Reachability with Annotations}
\label{sec:ReachabilityWithAnnotations}

The bounded reachability problem can be seen as a reachability problem with
annotations at the transitions. The value associated with a path can be computed from the
sequence of annotations -- in our scenario the counter operations. When
computing the transitive closure, we not only have to calculate basic 
reachability but also the (combined) annotations of the paths among the nodes.

In the following, we introduce a general algorithm to compute reachability with
annotations on prefix replacement systems. This procedure is based on the
widely known saturation principle. It creates an output which can directly be
transformed into a synchronous transducer calculating the annotations for paths
between a pair of system configurations. First, we formally describe the
requirements for an annotation domain. Subsequently, we explain the actual
saturation procedure and show how the reachability-cost problem for \RPRS{} can
be represented in the developed framework.

\subsection{Annotation Domains}

An annotation domain which is compatible with saturation has to satisfy certain
requirements. First, the concatenation operation on the annotations has to be
associative because the order in which paths in the pushdown system are
combined during saturation must not be important. Second, the termination of
saturation normally results from the fact that there are only finitely many
possible transitions in a finite automaton which can be added. However, this
argument does not suffice anymore if additional annotations from a potentially
infinitely large domain are considered. Since there may exist infinitely many
paths which are all annotated differently, this problem is inherent to the
considered question. Motivated by the search for \emph{good} or \emph{cheap}
paths in the context of the bounded reachability problem, we equip the
annotations with a partial order. This order has to be well-founded and must
not contain infinite anti-chains. Such an order is often called well-partial
order or partial well-order in the literature. Finally, we need to be able to
reverse paths because the saturation procedure we adapt operates in both directions
(predecessors and successors) simultaneously. Consequently, we require a
compatible reverse operation on the annotations. We formalize these requirements in the
form of \emph{well-partially ordered monoids with involution}.

\begin{defi}[Well-Partially Ordered Monoid with Involution]
	A well-partially ordered monoid with involution $\annotationMonoidM$
    is a 5-tuple  $(M,\circ,\neutralM,\le,\rev)$ where:
	\begin{enumerate}[label=(\roman*)]
		\item $(M,\circ,\neutralM)$ is a monoid
		\item $\le$ is a well-partial order. The order is compatible with the
           monoid operation, i.e., $a \le a' \wedge b \le b' \Implies a \circ b
            \le a' \circ b'$.
		\item $\rev: M \to M$ is the so called \emph{reverse} function. It is
an involution, i.e., $\rev(\rev(a)) = a$. Moreover, it satisfies the functional
equation $\rev(a \circ b) = \rev(b) \circ \rev(a)$.
		\item Order and the reverse function are compatible, i.e., $a \le a'
\Implies \rev(a) \le \rev(a')$ 
	\end{enumerate}
\end{defi}

\noindent We present a concrete example for such a structure later. In
Section~\ref{subsec:CostReachability}, we show how we can use this
formalism to obtain a quite natural and concise representation of sequences of B-counter operations. 

Before constructing specific instances of well-partially ordered monoids with involution,
we exhibit some general properties. Direct products are a useful tool
in many automaton constructions. Hence, we verify that the direct product of two
well-partially ordered monoids with involution is well-defined. This is
mostly formal checking that the properties are satisfied. 

\begin{lem}\label{lem:ProductOfRIOrderedMisRIOrderedM}
	Let $\annotationMonoidM_1 = (M_1,\circ_1,\neutralM[1],\le_1,\rev_1)$ and
$\annotationMonoidM_2 = (M_2,\circ_2,\neutralM[2],\le_2,\rev_2)$ be two
well-partially ordered monoids with involution. We define the
\emph{direct product} $\annotationMonoidM$ of these monoids in the usual way 
\[ \annotationMonoidM = (M_1 \times
M_2,\circ,(\neutralM[1],\neutralM[2]),\le,\rev)
\]
  with component-wise application of the operation $\circ$, the
involution $\rev$ and the compo\-nent-wise order $\le$, i.e., $ (m_1,m_2) \le
(m_1',m_2') \defEquiv m_1 \le_1 m_1' \wedge m_2 \le m_2'$.

  The monoid $\annotationMonoidM$ is a well-partially-ordered monoid with
involution.
\end{lem}
\begin{proof}
  It is clear that $\annotationMonoidM$ is a monoid. The compatibility of the
reverse function and the order are easy to check.
  In \cite{Higman52}, G. Higman showed that the component-wise order on the
product of two well-partial orders is also a well-partial order.
\end{proof}

We now define a general form of prefix replacement systems whose replacement
rules are annotated with elements from a well-partially ordered monoid with
involution. With the above remark that we can encode sequences of 
B-counter operations in a well-partially ordered monoid, the following 
definition can be seen as a generalization of counter automata.

\begin{defi}[Monoid Annotated Prefix Replacement System]\label{def:ResourcePRS}\ 

Let $\annotationMonoidM = (M,\circ,\neutralM,\le,\rev)$ be a well-partially ordered
  monoid with involution.
  A monoid annotated prefix replacement system is a triple $\prsR =
(\Sigma,\Delta,\annotationMonoidM)$ consisting of a finite alphabet $\Sigma$, a
finite transition relation $\Delta$ and a well-partially ordered monoid with
involution $\annotationMonoidM$.
  The prefix replacement rules in the transition relation $\Delta \subseteq
\Sigma^+ \times \Sigma^* \times M$ are annotated with elements from the monoid.
We also write $\pdsrule{u}{m}{v}$ for a prefix replacement rule $(u,v,m) \in
\Delta$.

 Let $w_1,w_2$ be two configurations. We say $w_2$ is an $m$-successor of $w_1$
and write $w_1 \configstepCost{m} w_2$ if $\exists (u,v,m) \in \Delta \;\exists
x \in \Sigma^*: w_1 = ux \wedge w_2 = vx$. Let $w_1, \ldots, w_n$ be a sequence
of configurations such that $w_i \configstepCost{m_i} w_{i+1}$ for all
$i = 1, \ldots, n-1$. We write $w_1 \configstepsCost{m} w_n$ with $m = m_1 \circ
\ldots \circ m_n$.
\end{defi}

\subsection{Annotation Aware Saturation}
\label{subsec:AnnotationAwareSat}

It is already known for quite some time that saturation methods can be used to
calculate the point-to-point reachability relation of pushdown- or prefix
replacement systems. In addition, it is also known that it is possible to
construct a synchronous transducer which computes this reachability relation.
This shows that the configuration space of a prefix replacement system with the
reachability relation is an automatic structure. One approach for such a
construction can be found in \cite{gtt-saturation}. Although this 
algorithm is for ground term replacement systems, it is easy to see that a
prefix replacement system is just a special case. We adapt the algorithm for
this special case and extend it to fit our needs.

The algorithm performs a two-sided saturation. Our adapted variant operates
over two $\eps$-NFAs that share some of their states. These two automata read
the changed prefixes of the two configurations. The common suffix is
ignored\footnote{This definition comes from the algorithm's origin as ground
tree replacement algorithm. There, the two subtrees which are framed by the
common (tree-)context are read by the automata.}. A pair $(uw,vw)$ of
configurations with common suffix $w$ is accepted by the pair of automata if
there is a run of the first automaton on $u$ and a run of the second automaton
on $v$ such that both end in the same (shared) state. The saturation algorithm
starts with two NFAs that recognize one rewrite step of a given prefix
replacement system. Then, subsequently, new $\eps$-transitions which each
simulate one or more prefix replacement steps are added in both automata. The
major reason for this two-sided construction is the possibility of very
different word lengths on both sides of a prefix replacement rule. In contrary
to usual pushdown systems, the left-hand side of a rule in prefix replacement
systems may be much longer than the right-hand side. 

Our adaptation extends the original algorithm to keep track of the annotations. We
implement this by annotating the transitions of the $\eps$-NFAs with elements
from the annotation domain of the prefix replacement system. Formally, we
define those NFAs by:
\begin{defi}[Monoid annotated $\eps$-NFA]
A monoid annotated $\eps$-NFA is a tuple $\automatonA =
(Q,\Sigma,\In,\Fin,\Delta,\annotationMonoidM)$. The components
$Q,\Sigma,\In,\Fin$ are as in usual NFAs. $\annotationMonoidM$ is a
well-partially ordered monoid with reverse-function. The finite transition
relation $\Delta$ is annotated with elements from $\annotationMonoidM$. It 
has the form $\Delta \subseteq Q \times (\Sigma \cup \{\eps\}) \times Q \times
M$. 
  
Each run of the automaton on a word $w \in \Sigma^*$ naturally defines a
monoid element. The value $m$ of a run is defined by the concatenation of the
values of the used transitions along the run. Consequently, the (finite) set of all possible
accepting runs of the automaton on the word $w$ induces a set of monoid elements
$S_w \subseteq M$. Additionally, we write $\automatonA: q_0 \apathS[m]{w} q$ to
indicate that there exists a run from $q_0$ to $q$ on the word $w$ with
accumulated annotation $m$.
\end{defi}

\begin{algorithm}[t]
 \SetKwInOut{Input}{input}\SetKwInOut{Output}{output}
 \SetKw{Continue}{continue}
 \SetKwFor{Find}{Find}{}{}
 \SetKwIF{If}{ElseIf}{Else}{if}{then}{else if}{else}{endif}
 \Input{monoid annotated prefix replacement system $\prsR =
(\Sigma,\DeltaR,\annotationMonoidM)$}
 \Output{monoid annotated automata $\automatonA^*_1 =
(Q_1,\Sigma,\{q_{1,\eps}\},\emptyset,\Delta^*_1)$ and $\automatonA^*_2 =
(Q_2,\Sigma,\{q_{2,\eps}\},\emptyset,\Delta^*_2)$}

  $\ell := \max\left\{ |u|,|v| \mid (u,v,m) \in \DeltaR \right\}$ 
  
   $\automatonA^0_1 := (Q_1,\Sigma,\{q_{1,\eps}\},\emptyset,\Delta^0_1)$ and
$\automatonA^0_2 := (Q_2,\Sigma,\{q_{2,\eps}\},\emptyset,\Delta^0_2)$ with
	$Q_{\mathrm{shared}} := \{ q_{(u,v,m)} \mid (u,v,m) \in \DeltaR \}$,
	$Q_1 := \{ q_{1,v} \mid v \in \Sigma^*, |v| \le \ell\} \cup
Q_{\mathrm{shared}}$
	$Q_2 := \{ q_{2,v} \mid v \in \Sigma^*, |v| \le \ell\} \cup
Q_{\mathrm{shared}}$ 
	$\Delta_{C,i} := \{ (q_{i,u},a,q_{i,v},\neutralM) \mid u,v \in \Sigma^*, a
\in \Sigma : v= ua \wedge |v| \le \ell\}$
	$\Delta^0_1 := \Delta_{C,1} \cup \{ (q_{1,u},\eps,q_{(u,v,m)},\neutralM)
\mid (u,v,m) \in \DeltaR \}$
	$\Delta^0_2 := \Delta_{C,2} \cup \{ (q_{2,v},\eps,q_{(u,v,m)},\neutralM)
\mid (u,v,m) \in \DeltaR \}$
  
  $i := 0$
  
  \While{automata can be updated}{
	$i := i +1$
	
	\Find{$w \in \Sigma^{\le \ell}$ and states $q, q_{(u,v,\bar m)}$ s.t.
$\automatonA^i_2: q_{2,\eps} \apathS[\rr]{w} q_{(u,v,\bar m)}$,
$\automatonA^i_1: q_{1,\eps} \apathS[\rl]{w} q$}{
	  $m_u := \bar m \rev(\rr)\rl$
	  
	  \uIf{$\exists (q_{(u,v,\bar m)},\eps,q,m_o) \in \Delta^i_1$ with $m_o >
m_u$}{
		  $\Delta^{i+1}_1 := \Delta^{i}_1 \setminus \{ (q_{(u,v,\bar
m)},\eps,q,\tilde m) \in \Delta^i_1 \mid \tilde m > m_u \}  \cup \{
(q_{(u,v,\bar m)},\eps,q,m_u)\}$
		}
	  \ElseIf{$\neg(\exists (q_{(u,v,\bar m)},\eps,q,m_o) \in \Delta^i_1 \text{
with } m_o \le m_u)$}{
		  $\Delta^{i+1}_1 := \Delta^i_1 \cup \{ (q_{(u,v,\bar
m)},\eps,q,m_u) \}$
		}
	  $\automatonA^{i+1}_2 := \automatonA^{i}_2$, $\automatonA^{i+1}_1 :=
(Q_1,\Sigma,\{q_{1,\eps}\},\emptyset,\Delta^{i+1}_1)$
	  
	  \Continue{}%
	}%
	\Find{$w \in \Sigma^{\le \ell}$ and states $q,q_{(u,v,\bar m)}$ s.t.
$\automatonA^i_2: q_{2,\eps} \apathS[\rr]{w} q$, $\automatonA^i_1: q_{1,\eps}
\apathS[\rl]{w} q_{(u,v,\bar m)}$}{
	  $m_u := \rev(\bar m) \rev(\rl)\rr$
	  
	  \uIf{$\exists (q_{(u,v,\bar m)},\eps,q,m_o) \in \Delta^i_2$ with $m_o >
m_u$}{
		  $\Delta^{i+1}_2 := \Delta^{i}_2 \setminus \{ (q_{(u,v,\bar
m)},\eps,q,\tilde m) \in \Delta^i_2 \mid \tilde m > m_u \}  \cup \{
(q_{(u,v,\bar m)},\eps,q,m_u)\}$
		}
	  \ElseIf{$\neg(\exists (q_{(u,v,\bar m)},q,m_o) \in \Delta^i_2 \text{ with
} m_o \le m_u)$}{
		  $\Delta^{i+1}_2 := \Delta^i_2 \cup \{ (q_{(u,v,\bar
m)},\eps,q,m_u) \}$
		}
	  $\automatonA^{i+1}_1 := \automatonA^{i}_1$, $\automatonA^{i+1}_2 :=
(Q_2,\Sigma,\{q_{2,\eps}\},\emptyset,\Delta^{i+1}_2)$
	  
	  \Continue{}%
	}%
  }%
  \KwResult{$\automatonA^*_1 := \automatonA^i_1$, $\automatonA^*_2 :=
\automatonA^i_2$}
  \caption{Two-sided saturation procedure\label{algo:GTTSaturation}}
\end{algorithm}

\begin{figure}[t]
	\begin{center}
		\begin{tikzpicture}
  \begin{scope}[shift={(0,0)}]
	\node[anchor=base] (reprule) at (0,1.5) {$\pdsrule{u}{\bar m}{v} \in \DeltaR$};
	\node (reprelation) at (0,0) {$u \configstep[\bar m] v$};
	\draw[-implies,double equal sign distance,double,thick] (reprule) -- (reprelation);
  \end{scope}
  
  \begin{scope}[shift={(4,0)}]
	\node[dotstyle] (WBegin) at (1,1.5) {};
	\node[left=0.7cm] at (WBegin) {$\mathfrak A_2^{i}:$};
	\node[left] at (WBegin) {$q_{2,\eps}$};
	\node[dotstyle] (WEnd) at (3.5,1.5) {};
	\node[right] at (WEnd) {$q_{(u,v,\bar m)}$};
	
	\draw[automatapath] (WBegin) -- (WEnd) 
		node[midway,above=1mm] {$w$} 
		node[midway,below=1mm] (w2Annot) {$\rr$}
		coordinate[midway] (pathmid);
	
	\path let \p1 = (pathmid) in node (autrelation) at (\x1,0) {$v \configstepsCost{\rev(\rr)} w$};
	
	\draw[-implies,double equal sign distance,double,thick] (w2Annot) -- (autrelation);
  \end{scope}
  
  \begin{scope}[shift={(1 ,-4)}]
	\node[dotstyle] (AnonBegin) at (0,2) {};
	\node[left=0.7cm] at (AnonBegin) {$\mathfrak A_1^{i+1}:$};
	\node[left] at (AnonBegin) {$q_{1,\eps}$};
	\node[dotstyle] (AnonEnd) at (3,2) {};
	\node[above] at (AnonEnd) {$q_{(u,v,\bar m)}$};
	
	\draw[automatapath] (AnonBegin) -- (AnonEnd) node[midway,above=1mm] {$u$};
	
	\node[dotstyle] (WA1Begin) at (0,1) {};
	\node[left=0.7cm] at (WA1Begin) {$\mathfrak A_1^{i}:$};
	\node[left] at (WA1Begin) {$q_{1,\eps}$};
	\node[dotstyle] (WA1End) at (3.5,1) {};
	\node (WA1Ext) at (4.5,1) {};
	\node[below] at (WA1End) {$q$};
	
	\draw[automatapath,>=] (WA1Begin) -- (WA1End) node[midway,above=1mm] {$w$} node[midway,below=1mm] (w1Annot) {$\rl$};
	\draw[automatapath,dotted] (WA1End) -- (WA1Ext);
	
	\draw[->,red] (AnonEnd) -- (WA1End) node[midway,right=2mm,anchor=base] {$\eps,$} 
		node[midway,right=6mm,anchor=base] (epsAnnotM) {$\bar m$} 
		node[midway,right=15mm,anchor=base] (epsAnnotMR) {$\rev(\rr)$} 
		node[midway,right=26mm,anchor=base] (epsAnnotML) {$\rl$};
  \end{scope}

	\draw[->] (reprelation.south) .. controls +(-90:2cm) and +(90:2.3cm) .. (epsAnnotM);
	\draw[->] (autrelation.south) to[out=-90,in=90] (epsAnnotMR);
	\draw[->] (w1Annot.east) to[out=-20,in=-90] (epsAnnotML);

 \end{tikzpicture}
	\end{center}
	\caption{Illustration of the idea of Algorithm~\ref{algo:GTTSaturation}}
	\label{fig:SaturationIdea}
\end{figure}
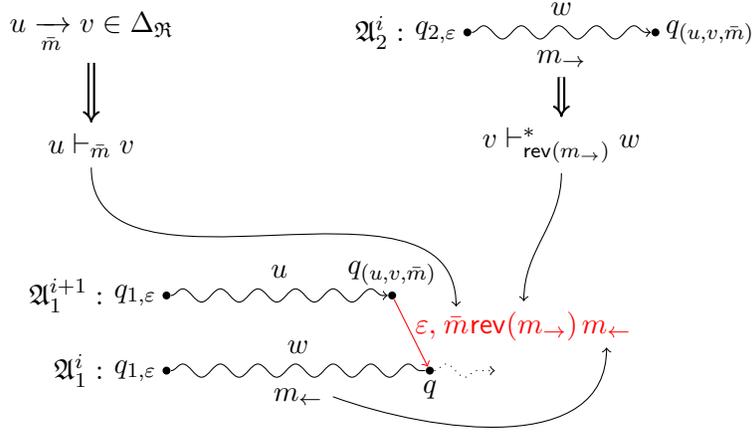

We describe the intuitive idea of Algorithm~\ref{algo:GTTSaturation} by first
explaining the goal of the algorithm followed by the explanation of one
saturation step in the first of the two automata. We remarked already that the
two automata start with recognizing the successor relation of prefix
replacement. Formally, this means that for every prefix replacement rule
$\pdsrule{u}{\bar m}{v} \in \DeltaR$, there is a (shared) final state
$q_{(u,v,\bar m)}$ in $\automatonA^0_1$ and $\automatonA^0_2$. Moreover, there
are runs $\automatonA^0_1: q_{1,\eps} \apathS[\neutralM]{u} q_{(u,v,\bar m)}$
and $\automatonA^0_2: q_{2,\eps} \apathS[\neutralM]{v} q_{(u,v,\bar m)}$. During
the saturation procedure, we want to add $\eps$-transitions which enable the
automata to simulate one or more prefix replacement steps on their own, such
that there is a run $\automatonA^*_1: q_{1,\eps} \apathS[\rl]{u'} q_{(u,v,\bar
m)}$ iff $u' \configstepsCost{\rl} u$ and symmetrically (but reversed) also that there is a run
$\automatonA^*_2: q_{2,\eps} \apathS[\rr]{v'} q_{(u,v,\bar m)}$ iff $v
\configstepsCost{\rev(\rr)} v'$. Note that this especially means that both runs
imply $u' \configstepsCost{\rl \bar m \rev(\rr)} v'$. 

In order to enable $\automatonA^i_1$ to simulate one more application of the
replacement rule $\pdsrule{u}{\bar m}{v}$, the algorithm uses the following
strategy, which is illustrated in Figure~\ref{fig:SaturationIdea}. First, it
finds a word $w \in \Sigma^{\le \ell}$ and a fitting run $\automatonA^i_2: q_{2,\eps}
\apathS[\rr]{w} q_{(u,v,\bar m)}$ (notice that $v = w$ is always possible). By
our intuition this means $v \configstepsCost{\rev(\rr)} w$. Subsequently, the
algorithm searches for a run $\automatonA^i_1: q_{1,\eps} \apathS[\rl]{w} q$ and
adds a transition from $q_{(u,v,\bar m)}$ to $q$ in $\automatonA^{i+1}_1$. After reading $u$, the
automaton can use the new $\eps$-transition and is now in a state as if it would
have read $w$ in the first place. This captures the sequence of replacement
operations $u \configstep[\bar m] v \configstepsCost{\rev(\rr)} w$, which is 
possible because $(u,v,\bar m) \in \DeltaR$ for $q_{(u,v,\bar m)}$ by construction.
Thereby, the automaton $\automatonA^{i+1}_1$ can now simulate one execution of this replacement sequence on its own. 

In addition to the correct operation of the replacement rules, the algorithm
also keeps track of the annotation. We saw that the $\eps$-transition simulates
replacement rules with accumulated annotation $\bar m\rev(\rr)$. Additionally, it
skips the part of reading $w$ on $\automatonA^i_1$, which also contains some
annotation $\rl$. Since $\automatonA^i_1$ would have read this annotation after
the replacements were made,  we also have to include $\rl$ into the
$\eps$-transition. For this reason, the complete annotation of the
$\eps$-transition to add is $\bar m\rev(\rr)\rl$. However, we only add this
transition if there is not already a transition with the same source and
destination and a better ($\le$-smaller) annotation. This ensures that we only
have finitely many transitions in the automaton although there might be
replacement paths with arbitrarily many different annotations for a pair of
configurations. Additionally, we remove all similar $\eps$-transitions with a
larger annotation for the same reason. 

The saturation in $\automatonA^i_2$ is symmetric with the roles of
$\automatonA^i_1$ and $\automatonA^i_2$ exchanged. However, the computation of
the annotation is a mirror-image of the computation in $\automatonA^i_2$
because $\automatonA^i_2$ simulates the replacement steps backwards. 

The result of the algorithm is a pair of automata which recognizes the
reachability relation and also contains all minimal annotations for paths from
the first to the second configuration. In the following, we provide formal
arguments for termination and correctness of the algorithm. This includes a
detailed analysis of the saturation steps and the book-keeping of the
annotations.

\begin{rem}
	Algorithm~\ref{algo:GTTSaturation} terminates for every input.
\end{rem}
\begin{proof}
	There are only finitely many pairs of states in $Q_1$ and $Q_2$.
Moreover, there are only two possibilities where a pair of states is considered
several times in the algorithm. First, the transitions can be updated because a pair of runs
leading to a smaller annotation was found. However, this can occur only finitely
many times since the order on the annotation domain is well-founded. Second,
several transitions can be added because the annotations are incomparable. This
can occur only finitely often, too. Otherwise this yields a (size)
increasing sequence of anti-chains. The union of all sets in the sequence would
be an infinite anti-chain in contradiction to the definition of the annotation
domain. Hence, the algorithm terminates after a finite number of steps.
\end{proof}

\begin{lem}\label{lem:CorrectnessOfGTTSaturation}
  Let $\prsR$ be a monoid annotated prefix replacement system, $\automatonA^*_1$
and $\automatonA^*_2$ the result of Algorithm~\ref{algo:GTTSaturation}. For two
configurations $w_1,w_2 \in \Sigma^*$ with $w_1 \ne w_2$ the following holds:
  \begin{enumerate}[label=(\roman*)]
   \item If $w_1 \configstepsCost{m} w_2$, then there are runs $\automatonA^*_1:
q_{1,\eps} \apathS[\rl]{w_1'} q_{(u,v,\bar m)}$ and $\automatonA^*_2: q_{2,\eps}
\apathS[\rr]{w_2'} q_{(u,v,\bar m)}$ for some $(u,v,\bar m) \in \DeltaR$ and $x
\in \Sigma^*$ such that $w_1 = w_1'x$, $w_2 = w_2'x$  and $\rl \bar m \rev(\rr)
\le m$.
	\item If there is a run $\automatonA^*_1: q_{1,\eps} 
	\apathS[\rl]{w_1} q_{(u,v,\bar m)}$ for some $(u,v,\bar m) \in \DeltaR$, 
	then $w_1z \configstepsCost{\rl} uz$ for all $z \in \Sigma^*$.
	\item If there is a run $\automatonA^*_2: q_{2,\eps} 
	\apathS[\rr]{w_2} q_{(u,v,\bar m)}$ for some $(u,v,\bar m) \in \DeltaR$, 
	then $vz \configstepsCost{\rev(\rr)} w_2z$ for all $z \in \Sigma^*$.
  \end{enumerate}
\end{lem}
\begin{proof} \ 
	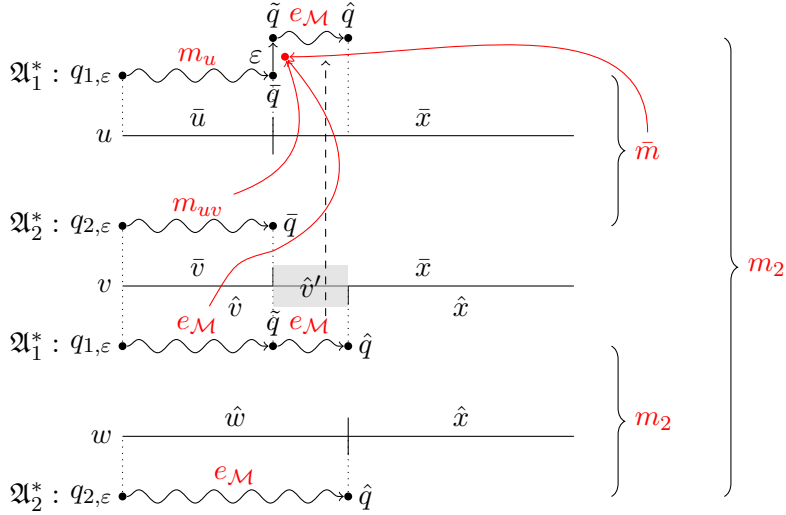
\begin{figure}[t]
		\begin{center}
			\begin{tikzpicture}
  \begin{scope}[shift={(0,4)}]
	\draw (0,0) -- (6,0);
	\node[left] (u) at (0,0) {$u$};
	\draw (2,-0.25) -- (2,0.25);
	\node[above] (uBar) at (1,0) {$\bar u$};
	\node[above] (xBarU) at (4,0) {$\bar x$};

	\node[dotstyle] (A1StartU) at (0,0.8) {};
	\draw[dotted] (A1StartU) -- (0,0);
	\node[left] at (A1StartU) {$q_{1,\eps}$};
	\node[left=0.7cm] at (A1StartU) {$\mathfrak A_1^*:$};

	\node[dotstyle] (A1EndU) at (2,0.8) {};
	\node[below] at (A1EndU) {$\bar q$};
	\draw[dotted] (A1EndU) -- (2,0);

	\draw[automatapath] (A1StartU) -- (A1EndU) node[midway,above,color=red] {$m_u$};

	\node[dotstyle] (A1ExtStartU) at (2,1.3) {};
	\node[above] at (A1ExtStartU) {$\tilde q$};
	
	\draw[->] (A1EndU) -- (A1ExtStartU) node[midway,left,color=black] {$\eps$} node[midway,right=1mm,dotstyle,color=red] (EpsTransitionResources) {};

	\node[dotstyle] (A1ExtEndU) at (3,1.3) {};
	\node[above] at (A1ExtEndU) {$\hat q$};
	\draw[dotted] (A1ExtEndU) -- (3,0);

	\draw[automatapath] (A1ExtStartU) -- (A1ExtEndU) coordinate[midway,below=0.3cm] (A1ExtMid) node[midway,above,color=red] {$\neutralM$};

  \end{scope}

  \begin{scope}[shift={(0,2)}]
	\draw (0,0) -- (6,0);
	\node[left] (v) at (0,0) {$v$};
	\draw (2,0) -- (2,0.25);
	\draw (3,-0.25) -- (3,0);
	\node[above] (vBar) at (1,0) {$\bar v$};
	\node[below] (vHat) at (1.5,0) {$\hat v$};
	\node[above] (xBarV) at (4,0) {$\bar x$};
	\node[below] (xHatV) at (4.5,0) {$\hat x$};

	\node[dotstyle] (A2StartV) at (0,0.8) {};
	\draw[dotted] (A2StartV) -- (0,0);
	\node[left] at (A2StartV) {$q_{2,\eps}$};
	\node[left=0.7cm] at (A2StartV) {$\mathfrak A_2^*:$};
  
	\node[dotstyle] (A1StartV) at (0,-0.8) {};
	\draw[dotted] (A1StartV) -- (0,0);
	\node[left] at (A1StartV) {$q_{1,\eps}$};
	\node[left=0.7cm] at (A1StartV) {$\mathfrak A_1^*:$};

	\node[dotstyle] (A2EndV) at (2,0.8) {};
	\node[right] at (A2EndV) {$\bar q$};
	\draw[dotted] (A2EndV) -- (2,0);
	\node[dotstyle] (A1EndV) at (3,-0.8) {};
	\draw[dotted] (A1EndV) -- (3,0);
	\node[right] at (A1EndV) {$\hat q$};
	\node[dotstyle] (A1MidV) at (2,-0.8) {};
	\draw[dotted] (A1MidV) -- (2,0);
	\node[above] at (A1MidV) {$\tilde q$};

	\draw[automatapath] (A2StartV) -- (A2EndV) node[midway,above,color=red] (Ruv) {$m_{uv}$};
	\draw[automatapath] (A1StartV) -- (A1MidV) node[midway,above,color=red] (Rwv) {$\neutralM$};
	\draw[automatapath] (A1MidV) -- (A1EndV) coordinate[midway,above=0.3cm] (A1RegMid) node[midway,above,color=red] {$\neutralM$};
	
	\node[anchor=west,minimum width=0.99cm,fill=black!20,fill opacity=0.5,text opacity=1] (vHatPrime) at (2,0) {$\hat v'$};
  \end{scope}

  \draw[->,dashed] ($(A1RegMid) + (0.2,0.1)$) -- ($(A1ExtMid) + (0.2,0)$);

  \begin{scope}[shift={(0,0)}]
	\draw (0,0) -- (6,0);
	\node[left] (w) at (0,0) {$w$};
	\draw (3,-0.25) -- (3,0.25);
	\node[above] (wHat) at (1.5,0) {$\hat w$};
	\node[above] (xHatW) at (4.5,0) {$\hat x$};

	\node[dotstyle] (A2StartW) at (0,-0.8) {};
	\draw[dotted] (A2StartW) -- (0,0);
	\node[left] at (A2StartW) {$q_{2,\eps}$};
	\node[left=0.7cm] at (A2StartW) {$\mathfrak A_2^*:$};

	\node[dotstyle] (A2EndW) at (3,-0.8) {};
	\draw[dotted] (A2EndW) -- (3,0);
	\node[right] at (A2EndW) {$\hat q$};

	\draw[automatapath] (A2StartW) -- (A2EndW) node[midway,above,color=red] {$\neutralM$};
  \end{scope}

  \draw[decorate,decoration={brace,amplitude=5pt}] (6.5,4.8) -- (6.5,2.8) node[midway,right=5pt,color=red] (Rbar) {$\bar m$};
  \draw[decorate,decoration={brace,amplitude=5pt}] (6.5,1.2) -- (6.5,-0.8) node[midway,right=5pt,color=red] {$m_2$};

  \draw[decorate,decoration={brace,amplitude=5pt}] (8,5.3) -- (8,-0.8) node[midway,right=5pt,color=red] {$m_2$};

  \draw[->,color=red] (Rbar) .. controls +(90:2cm) and +(0:2cm) .. (EpsTransitionResources);
  \draw[->,color=red] (Ruv) .. controls +(20:1.5cm) and +(-70:1cm) .. (EpsTransitionResources);
  \draw[->,color=red] (Rwv) .. controls +(60:1cm) .. +(1cm,1cm) 
							.. controls +(35:3cm) and +(-50:1cm) .. (EpsTransitionResources);
 \end{tikzpicture}
		\end{center}
		\caption{Illustration of the 1st case of part (i) of Lemma
	\ref{lem:CorrectnessOfGTTSaturation}}\label{fig:PartOneGTTSaturationProof}
  \end{figure}
  We first show (i) by induction on the number of replacement steps. Let 
  $\ell$ be as in the algorithm.
  
  \inductionstart{Let $w_1 \configstepCost{m} w_2$:}
	By definition of the successor relation, there is a replacement rule
$(u,v,\bar m) \in \DeltaR$ and a common suffix $x \in \Sigma^*$ such that $w_1 =
ux$, $w_2 = vx$ and $m = \bar m$. By definition of the automata
$\automatonA^0_1$ and $\automatonA^0_2$, which are included in $\automatonA^*_1$
and $\automatonA^*_2$, there are runs $\automatonA^*_1: q_{1,\eps}
\apathS[\neutralM]{u} q_{(u,v,\bar m)}$ and $\automatonA^*_2: q_{2,\eps}
\apathS[\neutralM]{v} q_{(u,v,\bar m)}$. Additionally, $\neutralM \bar m \rev(\neutralM) = \bar m$.
  \inductionstep{Let $u \configstepsCost[n+1]{m} w$:}
	By the definition of $\configstepsCost[n+1]{m}$, there is a $v$ such that $u
\configstepsCost[n]{m_1} v \configstepCost{m_2} w$ with $m_1m_2 = m$. By the
induction hypothesis, there is a common suffix $\bar x$ such that $u = \bar u
\bar x$, $v = \bar v \bar x$ and there are runs $\automatonA^*_1: q_{1,\eps}
\apathS[m_u]{\bar u} \bar q$ and $\automatonA^*_2: q_{2,\eps}
\apathS[m_{uv}]{\bar v} \bar q$ for some state $\bar q = q_{(w_1,w_2,\bar m)}$
such that $m_u\bar m\rev(m_{uv}) \le m_1$.
Additionally, there is a common suffix $\hat x$ such that  $v = \hat v \hat x$,
$w = \hat w \hat x$ and $(\hat v, \hat w, m_2) \in \DeltaR$. Now, distinguish
two cases depending on the length of $\hat x$ and $\bar x$.

	\textbf{1st case:} $|\hat x| \le |\bar x|$:
	\begin{indt}
	  The used words and runs in this case are shown in Figure
\ref{fig:PartOneGTTSaturationProof}.

	  One can write $v$ in the form $v = \bar v \hat v' \hat x$. With this
notation, we can also write $u = \bar u \hat v' \hat x$.

	  By construction of the automata $\automatonA^*_1$ and $\automatonA^*_2$,
there are two runs $\automatonA^*_1: q_{1,\eps} \apathS[\neutralM]{\bar v}
\tilde q \apathS[\neutralM]{\hat v'} \hat q$ and $\automatonA^*_2: q_{2,\eps}
\apathS[\neutralM]{\hat w} \hat q$ with $\hat q = q_{(\hat v, \hat w,m_2)}$. 
Additionally, we have $\hat v = \bar v \hat v'$ and thus 
$|\bar v| \le |\hat v| \le \ell$.

	  By the saturation algorithm, there is a transition $\bar q
\apath[m']{\eps} \tilde q$ in $\automatonA^*_1$ such that $m' \le \bar m
\rev(m_{uv})\neutralM$. Using this $\eps$-transition, one can create the
following run:
	  $$ \automatonA^*_1: q_{1,\eps} \apathS[m_u]{\bar u} \bar q
\apath[m']{\eps} \tilde q \apathS[\neutralM]{\hat v'} \hat q$$
	  So, this run and the run $\automatonA^*_2: q_{2,\eps}
\apathS[\neutralM]{\hat w} \hat q$ satisfy the first condition of the lemma. By
the induction hypothesis,  we have $m_u \bar m \rev(m_{uv}) \le m_1$. By the
monotonicity of the monoid operator, we have 
$$m_u m' m_2 \neutralM \le m_u \bar m \rev(m_{uv}) m_2 \le m_1m_2 = m$$

	\end{indt}
	\textbf{2nd case:} $|\hat x| > |\bar x|$:
	\begin{indt}
		This case is mostly analogous to the first one. The roles of
$\automatonA^*_1$ and $\automatonA^*_2$ are exchanged. One can write $w$ in the
form $w = \hat w \bar v' \bar x$ and $v$ in the form $v = \hat v \bar v' \bar
x$.

		By construction of the automata $\automatonA^*_1$ and $\automatonA^*_2$,
there are runs $\automatonA^*_1: q_{1,\eps} \apathS[\neutralM]{\hat v} \hat q$
and $\automatonA^*_2: q_{2,\eps} \apathS[\neutralM]{\hat w} \hat q$ with $\hat q
= q_{(\hat v, \hat w,m_2)}$. Furthermore, the inductively given run
$\automatonA^*_2: q_{2,\eps} \apathS[m_{uv}]{\bar v} \bar q$ can be divided into
two parts $\automatonA^*_2: q_{2,\eps} \apathS[m_{\hat v}]{\hat v} \tilde q
\apathS[m_{\bar v'}]{\bar v'} \bar q$.

		Since $|\hat v| \le \ell$, the the saturation algorithm guarantees that 
		there is a transition $\hat q \apath[m']{\eps} \tilde q$ in 
		$\automatonA^*_2$ such that $m' \le  \rev(m_2) 
		\rev(\neutralM)m_{\hat v} = \rev(m_2) m_{\hat v}$. 
		Using this $\eps$-transition, one can create the following run:
		$$ \automatonA^*_2: q_{2,\eps} \apathS[\neutralM]{\hat w} \hat q
\apath[m']{\eps} \tilde q \apathS[m_{\bar v'}]{\bar v'} \bar q$$
		So, this run and the run $\automatonA^*_1: q_{1,\eps} \apathS[m_u]{\bar
u} \bar q$ satisfy the first condition of the lemma. 
		Consequently, we obtain by the application of the functional equation of
$\rev$ and the compatibility with the order:
		\begin{align*}
		  m_u \bar m \rev(m'm_{\bar v'}) &= m_u \bar m \rev(m_{\bar v'})
\rev(m') \\
		  &\le   m_u \bar m \rev(m_{\bar v'}) \rev(\rev(m_2) m_{\hat v}) \\
		  &= m_u \bar m \rev(m_{\bar v'}) \rev(m_{\hat v})m_2 \\
		  &= m_u \bar m \rev(\underbrace{m_{\hat v}m_{\bar v'}}_{m_{uv}})m_2 \\
		  &\ihle m_1 m_2 = m
		\end{align*}
	\end{indt}
  \inductionend
  
\noindent	Now we show the parts (ii) and (iii) by induction on the number of
steps of the algorithm. Depending on the automaton in which the saturation step was
executed either the statement of part (ii) or of part (iii) needs to be proven.
Nevertheless, we need to prove both statements together because of the
interplay of both automata in the saturation procedure.

	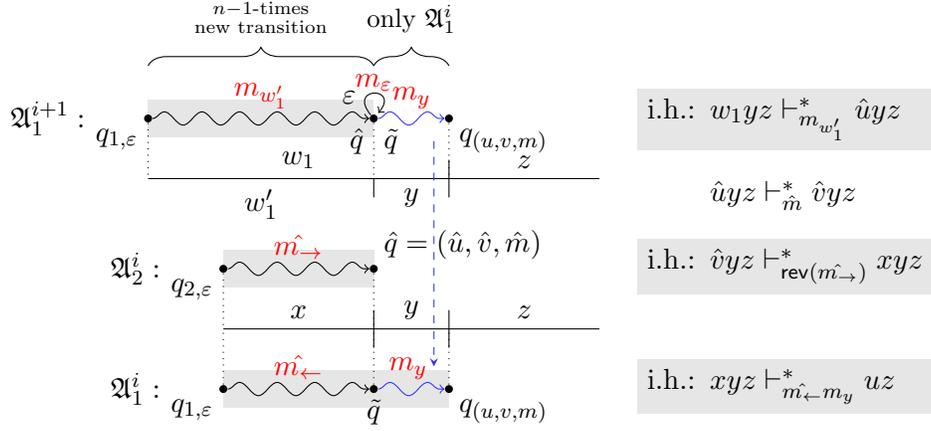
\begin{figure}[t]
		\begin{center}
			\begin{tikzpicture}
	\tikzstyle{graybox}=[fill=black!10]
	\tikzstyle{grayboxtext}=[anchor=base west,text width=3.75cm,align=left]
	\tikzstyle{highlightrun}=[color=blue!80]
  \begin{scope}[shift={(0,4)}]
	  \draw (0,0) -- (6,0);
	  \node[left] (u) at (0,0) {};
	  
	  \draw (3,-0.25) -- +(0,0.25);
	  \draw (4,-0.25) -- +(0,0.5);
		
	  \node[above] at (5,0) {$z$};
	  \node[above] at (2,0) {$w_1$};
	  \node[below] at (1.5,0) {$w_1'$};
	  \node[below] at (3.5,0) {$y$};

	  \node[dotstyle] (UBegin) at (0,0.8) {};
	  \draw[dotted] (UBegin) -- (0,0);
	  \node[left=0.7cm] at (UBegin) {$\mathfrak A_1^{i+1}:$};
	  \node[below left] at (UBegin) {$q_{1,\eps}$};

	  \node[dotstyle] (UMid) at (3,0.8) {} edge [out=130,in=50,loop] node[pos=0.3,left] {$\eps$} node[midway,above=-1mm,color=red] {$m_\eps$} ();
	  \draw[dotted] (UMid) -- (3,0);
	  \node[below left] at (UMid) {$\hat q$};
	  \node[below right] at (UMid) {$\tilde q$};
	  
	  \node[dotstyle] (UEnd) at (4,0.8) {};
	  \draw[dotted] (UEnd) -- (4,0);
	  \node[below right] at (UEnd) {$q_{(u,v,m)}$};

	  \draw[automatapath] (UBegin) -- (UMid) node[midway,above,color=red] {$m_{w_1'}$};
	  \draw[automatapath,highlightrun] (UMid) -- (UEnd) coordinate[midway,below=3mm] (UMidMarker) node[midway,above,color=red] {$m_y$};
	  
	  \begin{pgfonlayer}{background}
		\fill[graybox] ($(UBegin) - (0,0.25)$) rectangle ($(UMid) + (0,0.25)$);
	  \end{pgfonlayer}

	  \node[graybox,grayboxtext] at (6.5,0.8) {i.h.: $w_1yz \configstepsCost{m_{w_1'}} \hat u yz$}; 
	  
	  \node[grayboxtext] at (6.5,-0.3) {\hphantom{i.h.: }$\hat u yz \configstepsCost{\hat m} \hat v yz$};

	  \draw[curlybracket] ($(UBegin)+(0,0.7)$)--($(UMid)+(0,0.7)$) node[midway,above=3mm,font={\small}] {$\begin{smallmatrix} n-1\text{-times} \\ \text{new transition}\end{smallmatrix}$};
	  \draw[curlybracket] ($(UMid)+(0,0.7)$)--($(UEnd)+(0,0.7)$) node[midway,above=3mm,font={\small}] {only $\mathfrak A^i_1$};
	  
  \end{scope}

  \begin{scope}[shift={(0,2)}]
	  \draw (1,0) -- (6,0);
	  \node[left] (v) at (1,0) {};

	  \draw (3,-0.25) -- +(0,0.5);
	  \draw (4,-0.25) -- +(0,0.5);

	  \node[above] at (5,0) {$z$};
	  \node[above] at (2,0) {$x$};
	  \node[above] at (3.5,0) {$y$};

	  \node[dotstyle] (V2Begin) at (1,0.8) {};
	  \draw[dotted] (V2Begin) -- (1,0);
	  \node[dotstyle] (V2End) at (3,0.8) {};
	  \draw[dotted] (V2End) -- (3,0);
	  \node[left=0.7cm] at (V2Begin) {$\mathfrak A_2^{i}:$};
	  \node[below left] at (V2Begin) {$q_{2,\eps}$};
	  \node[above right] at (V2End) {$\hat q = (\hat u,\hat v, \hat m)$};

	  \draw[automatapath] (V2Begin) -- (V2End) node[midway,above,color=red] {$\hat\rr$};
	   \begin{pgfonlayer}{background}
		\fill[graybox] ($(V2Begin) - (0,0.25)$) rectangle ($(V2End) + (0,0.25)$);
	  \end{pgfonlayer}
	  
	  \node[graybox,grayboxtext] at (6.5,0.8) {i.h.: $\hat vyz \configstepsCost{\rev(\hat \rr)} xyz$}; 

	  \node[dotstyle] (V1Begin) at (1,-0.8) {};
	  \draw[dotted] (V1Begin) -- (1,0);
	  \node[dotstyle] (V1Mid) at (3,-0.8) {};
	  \draw[dotted] (V1Mid) -- (3,0);
	  \node[dotstyle] (V1End) at (4,-0.8) {};
	  \draw[dotted] (V1End) -- (4,0);
	  \node[left=0.7cm] at (V1Begin) {$\mathfrak A_1^{i}:$};
	  \node[below left] at (V1Begin) {$q_{1,\eps}$};
	  \node[below] at (V1Mid) {$\tilde q$};
	  \node[below right] at (V1End) {$q_{(u,v,m)}$};

	  \draw[automatapath] (V1Begin) -- (V1Mid) node[midway,above,color=red] {$\hat\rl$};
	  \draw[automatapath,highlightrun] (V1Mid) -- (V1End) coordinate[midway,above=3mm] (V1MidMarker) node[pos=0.45,above,color=red] {$m_y$};
	   \begin{pgfonlayer}{background}
		\fill[graybox] ($(V1Begin) - (0,0.25)$) rectangle ($(V1End) + (0,0.25)$);
	  \end{pgfonlayer}
	  
	  \node[graybox,grayboxtext] at (6.5,-0.8) {i.h.: $xyz \configstepsCost{\hat \rl m_y} uz$}; 
  \end{scope}

  \draw[->,dashed,highlightrun,>=stealth] ($(UMidMarker) + (3mm,0)$) -- ($(V1MidMarker) + (3mm,0)$);


\end{tikzpicture}
		\end{center}
		\caption{Illustration of part (ii) of Lemma
	\ref{lem:CorrectnessOfGTTSaturation}}\label{fig:PartTwoGTTSaturationProof}
	\end{figure}

	\inductionstart{Let $\automatonA^0_1: q_{1,\eps} \apathS[m_1]{w_1}
q_{(u,v,m)}$ or $\automatonA^0_2: q_{2,\eps} \apathS[m_2]{w_2} q_{(u,v,m)}$
for some $(u,v,m) \in \DeltaR$ and $z \in \Sigma^*$:}
	By the construction of the automata $\automatonA^0_1$ and $\automatonA^0_2$,
we have $w_1 = u$, $w_2 =  v$, $m_1 = \neutralM$, $m_2 = \neutralM$ and thus obtain
$w_1z \configstepsCost[0]{\neutralM} uz$ and $vz
\configstepsCost[0]{\neutralM} w_2z$ as desired.

  \inductionstep{Let $\automatonA^{i+1}_1: q_{1,\eps} \apathS[\rl]{w_1}
q_{(u,v,m)}$ for some $(u,w,m) \in \DeltaR$ and $z \in \Sigma^*$:}
	We assume that the new/updated transition $\hat q \apath[m_\eps]{\eps}
\tilde q$ is in $\automatonA^{i+1}_1$ (otherwise there is nothing to show here).
	Let $n \in \nat$ be the number of occurrences of the new transition in
$\automatonA^{i+1}_1: q_{1,\eps} \apathS[\rl]{w_1} q_{(u,v,m)}$.

	\inductionstart{Let $n = 0$:}
	If the run $\automatonA^{i+1}_1: q_{1,\eps} \apathS[\rl]{w_1} q_{(u,v,m)}$
does not contain the new transition, the claim follows directly by the induction
hypothesis of the outer induction. 
	\inductionstep{Let $n > 0$:}
	The used words and runs of the construction are shown in Figure
\ref{fig:PartTwoGTTSaturationProof}.

	The run on $\automatonA^{i+1}_1$ can be represented by:
	$$ \automatonA^{i+1}_1: \underbrace{q_{1,\eps} \apathS[m_{w_1'}]{w_1'} \hat
q}_{\begin{smallmatrix} \text{new trans. only} \\ \text{$n-1$ times}
\end{smallmatrix}} \qquad \hat q \apath[m_\eps]{\eps} \tilde q \qquad
\underbrace{\tilde q \apathS[m_y]{y} q_{(u,v,m)}}_{\text{only
}\automatonA^i_1}  \text{ with } w_1 = w_1'y$$
	By the saturation algorithm, there is a word $x \in \Sigma^*$ and a pair
of runs $\automatonA^i_1: q_{1,\eps} \apathS[\hrl]{x} \tilde q$ and
$\automatonA^i_2: q_{2,\eps} \apathS[\hrr]{x} \hat q$ with $\hat q = q_{(\hat
u, \hat v, \hat m)}$ such that $m_\eps = \hat m \rev(\hrr)\hrl$ which led
to the construction of the new transition.

	Since $\automatonA^{i+1}_1: q_{1,\eps} \apathS[m_{w_1'}]{w_1'} \hat q$ uses
the new transition only $n-1$ times we have $w_1z = w_1'yz
\configstepsCost{m_{w_1'}} \hat uyz$ by the inner induction hypothesis.
Furthermore, we obtain from the outer induction hypothesis (part (iii) of the
lemma) that $\hat vyz \configstepsCost{\rev(\hrr)}  xyz$.

	Additionally, it is possible to construct the following run with the
given run of $\automatonA^{i}_1$ on $w_1$ and the run of $\automatonA^{i}_1$
which led to the construction of the new transition:
	$$ \automatonA^{i}_1 : q_{1,\eps} \apathS[\hrl]{x} \tilde q
\apathS[m_y]{y} q_{(u,v,m)}$$

	This run does not use the new transition. Consequently, the outer
induction hypothesis yields $xyz \configstepsCost{\hrl m_y} uz$. In total:
	$$ w_1z = w_1'yz \configstepsCost{m_{w_1'}} \hat uyz
\configstepCost{\hat m} \hat vyz \configstepsCost{\rev(\hrr)}  xyz
\configstepsCost{\hrl m_y} uz$$
	with 
	$$ m_{w_1'} \underbrace{\hat m \rev(\hrr) \hat \rl}_{m_\eps} m_y =
m_{w_1'}m_\eps m_y = \rl$$
	\inductionend
	  \inductionstep{Let $\automatonA^{i+1}_2: q_{2,\eps} \apathS[\rr]{w_2}
q_{(u,v,m)}$ for some $(u,v,m) \in \DeltaR$ and $z \in \Sigma^*$:}
	  This step is mostly analogous to the previous case -- just the roles of the
two automata are exchanged. Thus, we now assume that the new/updated
transition $\hat q \apath[m_\eps]{\eps} \tilde q$ is in $\automatonA^{i+1}_2$.
	Let $n \in \nat$ be the number of occurrences of the new transition in
$\automatonA^{i+1}_2: q_{2,\eps} \apathS[\rr]{w_2} q_{(u,v,m)}$.

	\inductionstart{Let $n = 0$:}
	If the run $\automatonA^{i+1}_2: q_{2,\eps} \apathS[\rr]{w_2} q_{(u,v,m)}$
does not contain the new transition, the claim follows directly by the induction
hypothesis of the outer induction. 
	\inductionstep{Let $n > 0$:}
	The run on $\automatonA^{i+1}_2$ can be represented by:
	$$ \automatonA^{i+1}_2: \underbrace{q_{2,\eps} \apathS[m_{w_2'}]{w_2'} \hat
q}_{\begin{smallmatrix} \text{new trans. only} \\ \text{$n-1$ times}
\end{smallmatrix}} \qquad \hat q \apath[m_\eps]{\eps} \tilde q \qquad
\underbrace{\tilde q \apathS[m_y]{y} q_{(u,v,m)}}_{\text{only
}\automatonA^i_2}  \text{ with } w_2 = w_2'y$$
	By the saturation algorithm, there is a word $x \in \Sigma^*$ and a pair
of runs $\automatonA^i_1: q_{1,\eps} \apathS[\hrl]{x} \hat q$ and
$\automatonA^i_2: q_{2,\eps} \apathS[\hrr]{x} \tilde q$ with $\hat q =
q_{(\hat u, \hat v, \hat m)}$ such that $m_\eps = \rev(\hat m)
\rev(\hrl)\hrr$ which led to the construction of the new transition.

	Since $\automatonA^{i+1}_2: q_{2,\eps} \apathS[m_{w_2'}]{w_2'} \hat q$ uses
the new transition only $n-1$ times, the inner induction hypothesis yields $\hat v
yz \configstepsCost{\rev\left(m_{w_2'}\right)} w_2'yz$.
Furthermore, we obtain from the outer induction hypothesis (part (ii) of the
lemma) that $xyz \configstepsCost{\hat \rl} \hat uyz$.

	Additionally, it is possible to construct the following run with the
given run of $\automatonA^{i}_2$ on $w_2$ and the run of $\automatonA^{i}_2$
which led to the construction of the new transition:
	$$ \automatonA^{i}_2 : q_{2,\eps} \apathS[\hrr]{x} \tilde q
\apathS[m_y]{y} q_{(u,v,m)}$$

	This run does not use the new transition. Consequently, the outer
induction hypothesis yields $vz \configstepsCost{\rev(\hrr m_y)} xyz$. In
total:
	$$ vz \configstepsCost{\rev(\hrr m_y)} xyz \configstepsCost{\hrl}
\hat uyz \configstepsCost{\hat m} \hat v yz
\configstepsCost{\rev\left(m_{w_2'}\right)} w_2'yz$$
	with 
	\begin{align*}
		\rev(\hrr m_y) \hat \rl \hat m \rev\left(m_{w_2'}\right) &=
\rev(m_y) \rev(\hrr) \hrl \hat m \rev\left(m_{w_2'}\right) \\
		&= \rev(m_y)\rev( \underbrace{\rev(\hat m) \rev( \rl) \hrr}_{m_\eps}
)\rev\left(m_{w_2'}\right) \\
		&= \rev( m_{w_2'} m_\eps m_y) \\
		&= \rev( \rr )
	\end{align*}
	\inductionend
	\inductionend\vspace{-\baselineskip}
\end{proof}

\subsection{Cost-Reachability in \RPRS{}}
\label{subsec:CostReachability}

Although the previously presented algorithm shows that we are able to compute
the reachability relation of prefix replacement systems with annotations, there
are still two problems to settle in order to solve the problem stated in our
motivating question. First, we have to encode the counter operations in the form
of a well-partially ordered monoid with involution. The main problem here
is that evaluating (mapping partial sequences to maximal counter values)
sequences of counter operations is inherently not associative. Thus, we need to
capture the essential information contained in a sequence of counter
operations but in a more accessible representation. Second, we have to
construct a synchronous resource transducer from the result of the algorithm. 

We define \emph{counter profiles} to capture the behavior of counter sequences
and provide a way to store a sequence of counter operations such that an
associative concatenation can be defined. A counter profile is a triple whose
elements are either natural numbers or $\na$ (read ``n/a'') meaning \emph{not applicable}. 
We map a sequence of counter operations to a counter profile based on the following
ideas. If there is no reset in the sequence, we just store the number of
increments in the first component. All other components are set to $\na$.
If it contains a reset, we store the number of increments before the first reset 
in the first component and the number of increments after the last reset in the 
third component. The second component contains the maximal number of
increments between two subsequent resets. If there are less than two resets in
the sequence, this component remains $\na$. It is easy to see that such a
profile contains sufficient information to define an associative concatenation
operator and that the profile $(0,\na,\na)$ is neutral for this operator. We
denote the set of all valid counter profiles, i.e., those profiles resulting
from some counter sequence, by $\counterProfiles$ and the above described 
mapping from counter sequences to counter profiles by 
$\cprofile: \{\iOp,\rOp,\nOp\}^* \to \counterProfiles$.

\begin{prop}\label{prop:CounterProfilesFormAnnotationMonoid}
	The structure $(\counterProfiles,\circ,(0,\na,\na),\cple,\rev)$ is a
	well-partially ordered mon\-oid with involution where $\circ$ is the
	concatenation operator induced by the concatenation of counter sequences,
	$\cple$ is the component-wise order on the profiles. In each component the
	natural numbers are ordered canonically, the element $\na$ is incomparable
	to all numbers. The function $\rev$ is defined by:
\[
	\rev: \counterProfiles \to \counterProfiles, (\ipl,\cmax,\ipr) \mapsto
	\begin{cases} 
		(\ipl,\na,\na) & \text{if } \ipr = \na \\
		(\ipr,\cmax,\ipl) & \text{otherwise}
	\end{cases}
\]
Note that the first case corresponds to counter sequences without reset, which
are reverse-invariant.
\end{prop}
\begin{proof}
	By construction the structure $(\counterProfiles,\circ,(0,\na,\na))$ is a
monoid. Furthermore, one can verify by checking all cases that the order $\cple$
is compatible with $\circ$ and the reverse function. We show that the order
$\cple$ is well-founded and has only finite anti-chains by showing that it is
the product of three orders satisfying these conditions. This argumentation is
analogous to Lemma~\ref{lem:ProductOfRIOrderedMisRIOrderedM}. The canonical
order on $\nat$ with one additional element $\na$ which is incomparable to all
numbers is well-founded and has only finite anti-chains since every set of more
than one natural number forms a chain. Thus only two element sets containing one
number and $\na$ can be anti-chains. The order on counter profiles is the
product of three times this order.
\end{proof}

We remark that the framework of well-partially ordered monoids with involution
can be used to capture sequences of counters with S-automaton semantics as well.
The idea to construct S-counter profiles is very similar to B-counter profiles.
Analogously, they store the number of increments before the first and after
the last reset. They additionally store whether these resets were $\crOp$ or
$\rOp$. Following the semantics of S-automata, the profiles do not store the
maximal counter value between subsequent (check-)resets but the minimal ones. 
Due to the fact that there are $\crOp$ and $\rOp$ in S-semantics, they have 
to store this for both directions: when reading the counter sequence from 
left to right and when reading from right to left.

(B-)Counter profiles allow the translation of \RPRS{} into monoid annotated prefix
replacement systems with a direct relation between their semantics. This
translation transforms $\iOp$ into $(1,\na,\na)$, $\rOp$ into
$(0,\na,0)$ and $\nOp$ into the neutral element $(0,\na,\na)$. In this
translated system we define $u \configstepsLesserCost{k} v$ iff $u
\configstepsCost{m} v$ for some $m = (\ipl,\cmax,\ipr) \in \counterProfiles$
with $\max(\ipl,\cmax,\ipr) \le k$. We already explained that $\cprofile$ is a
\mbox{(monoid-)}homomorphism. Consequently, the relation
$\configstepsLesserCost{k}$ is preserved under the transformation of the \RPRS{}
to the annotated prefix replacement system. Additionally, this also holds the
other way around. For a pair $u \configstepsLesserCost{k} v$ in the profile
annotated prefix replacement system one can directly construct a path also
satisfying $u \configstepsLesserCost{k} v$ in the corresponding \RPRS{}.
Moreover, we remark that if $u \configstepsLesserCost{k} v$ holds for a pair of
configurations in the profile annotated prefix replacement system, then there is
a $\cple$-minimal profile which witnesses that. Formally, let $M$ be the set of
all those elements $m$ s.t. $u \configstepsCost{m} v$. There is a
$(\ipl,\cmax,\ipr) = m_{\downarrow} \in \min M$ s.t.\ $\max(\ipl,\cmax,\ipr) \le
k$ because $\cple$-smaller elements can only have a smaller $\max$ than the
$\max$ of all components. 

Analogously, we can transform counter profiles back into equivalent counter
sequences. We define the mapping $\profileToOp: (\ipl,\cmax,\ipl) \mapsto
\iOp^{\ipl}[\rOp\iOp^{\cmax}][\rOp\iOp^{\ipr}]$, where the parts in square
brackets are only needed if $\cmax$ or $\ipr$ are not $\na$, to convert a 
profile back into an equivalent sequence of counter operations. This 
transformation allows us to obtain normal B-automata in the end.

Even though the above description did not consider \RPRS{} with several
counters, we are not restricted to the single counter scenario. By Lemma
\ref{lem:ProductOfRIOrderedMisRIOrderedM} the direct product of two monoids as
presented in Proposition \ref{prop:CounterProfilesFormAnnotationMonoid} is
again a well-partially ordered monoid with reverse-function. Since all counters
are handled independently from each other in an \RPRS{}, the direct product of
$n$ copies of the counter profiles exactly corresponds to the behavior of an $n$
counter \RPRS{} for the same reasons as explained in the previous paragraph.
Accordingly, we have $u \configstepsLesserCost{k} v$ if the maximum of all
entries in all profiles is less than or equal to $k$. 

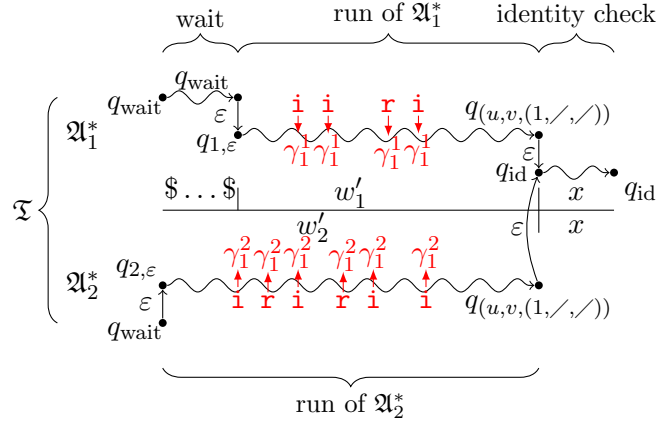
\begin{figure}[t]
  \begin{center}
  \begin{tikzpicture}
  \node[dotstyle] (A2Start) at (0,-0.5) {};
  \node[dotstyle] (A2Start2) at (0,0) {};
  \node[dotstyle] (A1Start) at (0,2.5) {};
  \node[dotstyle] (A1WaitEnd) at (1,2.5) {};
  \node[dotstyle] (A1WaitEnd2) at (1,2) {};
  \node[dotstyle] (A1NormEnd) at (5,2) {};
  \node[dotstyle] (A2NormEnd) at (5,0) {};

  \node[dotstyle] (CommonStart) at (5,1.5) {};
  \node[dotstyle] (CommonEnd) at (6,1.5) {};

  \node[above] at (A1NormEnd) {$q_{(u,v,(1,\na,\na))}$};
  \node[below] at (A2NormEnd) {$q_{(u,v,(1,\na,\na))}$};
  
  \node[below right] at (CommonEnd) {$q_{\mathrm{id}}$};
  \node[left] at (CommonStart) {$q_{\mathrm{id}}$};

  \draw[automatapath] (A1Start) -- (A1WaitEnd);
  \draw[curlybracket] ($(A1Start) + (0,0.6)$) -- ($(A1WaitEnd) + (0,0.6)$) node[midway,above=2mm,font={\small}] {wait}; 
  \draw[automatapath] (A1WaitEnd2) -- (A1NormEnd);
  \draw[curlybracket] ($(A1WaitEnd) + (0,0.6)$) -- ($(A1NormEnd) + (0,1.1)$) node[midway,above=2mm,font={\small}] {run of $\automatonA^*_1$}; 
 
  \draw[curlybracket] ($(CommonStart) + (0,1.6)$) -- ($(CommonEnd) + (0,1.6)$) node[midway,above=2mm,font={\small}] {identity check};

  \draw[automatapath] (A2Start2) -- (A2NormEnd);
  \draw[curlybracket] ($(A2NormEnd) + (0,-1.1)$) -- ($(A2Start) + (0,-0.6)$) node[midway,below=2mm,font={\small}] {run of $\automatonA^*_2$}; 

  \draw[automatapath] (CommonStart) -- (CommonEnd);

  \draw[->] (A2NormEnd) to[out=110,in=-110] node [midway,left=-1mm] {$\eps$}
(CommonStart);
  \draw[->] (A1NormEnd) -- node [midway,left=-1mm] {$\eps$} (CommonStart);

  \draw[-] (0,1) -- (6,1);
  \draw[-] (5,0.7) -- (5,1.3);
  \node[above] at (5.5,1) {$x$};
  \node[below] at (5.5,1) {$x$};

  \draw[-] (1,1) -- (1,1.3);
  \node[above] at (0.5,1) {$\$ \ldots \$$};
  \node[above] at (2.5,0.9) {$w_1'$};
  \node[below] at (2,1.11) {$w_2'$};
  
  \begin{scope}[>=latex,anchor=base]
	\node[color=red,inner sep=1pt] (i11) at (1.8,2.3) {$\iOp$};
	\node[color=red,below=4mm,inner sep=0pt] (g11) at (i11)  {$\gamma^1_1$};
	\draw[->,color=red] (i11) -- (g11);
	
	\node[color=red,inner sep=1pt] (i12) at (2.2,2.3) {$\iOp$};
	\node[color=red,below=4mm,inner sep=0pt] (g12) at (i12)  {$\gamma^1_1$};
	\draw[->,color=red] (i12) -- (g12);
	
	\node[color=red,inner sep=1pt] (r13) at (3,2.3) {$\rOp$};
	\node[color=red,below=4mm,inner sep=0pt] (g13) at (r13)  {$\gamma^1_1$};
	\draw[->,color=red] (r13) -- (g13);
	
	\node[color=red,inner sep=1pt] (i14) at (3.4,2.3) {$\iOp$};
	\node[color=red,below=4mm,inner sep=0pt] (g14) at (i14)  {$\gamma^1_1$};
	\draw[->,color=red] (i14) -- (g14);

	\node[color=red,inner sep=1pt] (i21) at (1,-0.3) {$\iOp$};
	\node[color=red,above=4mm,inner sep=0pt] (g21) at (i21)  {$\gamma^2_1$};
	\draw[->,color=red] (i21) -- (g21);
	\node[color=red,inner sep=1pt] (r22) at (1.4,-0.3) {$\rOp$};
	\node[color=red,above=4mm,inner sep=0pt] (g22) at (r22)  {$\gamma^2_1$};
	\draw[->,color=red] (r22) -- (g22);
	\node[color=red,inner sep=1pt] (i23) at (1.8,-0.3) {$\iOp$};
	\node[color=red,above=4mm,inner sep=0pt] (g23) at (i23)  {$\gamma^2_1$};
	\draw[->,color=red] (i23) -- (g23);
	
	\node[color=red,inner sep=1pt] (r24) at (2.4,-0.3) {$\rOp$};
	\node[color=red,above=4mm,inner sep=0pt] (g24) at (r24)  {$\gamma^2_1$};
	\draw[->,color=red] (r24) -- (g24);
	
	\node[color=red,inner sep=1pt] (i25) at (2.8,-0.3) {$\iOp$};
	\node[color=red,above=4mm,inner sep=0pt] (g25) at (i25)  {$\gamma^2_1$};
	\draw[->,color=red] (i25) -- (g25);
	
	\node[color=red,inner sep=1pt] (i26) at (3.5,-0.3) {$\iOp$};
	\node[color=red,above=4mm,inner sep=0pt] (g26) at (i26)  {$\gamma^2_1$};
	\draw[->,color=red] (i26) -- (g26);
  \end{scope}
  
  \node[left=0.7cm] (A1) at (0,2) {$\automatonA^*_1$};
  \node[left=0.7cm] (A2) at (0,0) {$\automatonA^*_2$};
  \coordinate[left=1.4cm] (Upper) at (A1Start);
  \coordinate[left=1.4cm] (Lower) at (A2Start);
  
  \node[below left=-2mm] at (A1Start) {$q_{\mathrm{wait}}$};
  \node[above left=-0.5mm] at (A1WaitEnd) {$q_{\mathrm{wait}}$};
  \node[below left=-2mm] at (A2Start) {$q_{\mathrm{wait}}$};
  
  \draw[->] (A2Start) -- (A2Start2) node[midway,left] {$\eps$};
  \node[above left=-1mm] at (A2Start2) {$q_{2,\eps}$};
  
  \draw[->] (A1WaitEnd) -- (A1WaitEnd2) node[midway,left] {$\eps$};
  \node[below left=-2mm] at (A1WaitEnd2) {$q_{1,\eps}$};

  \draw[curlybracket] (Lower) -- (Upper) node[midway,left=2mm] {$\automatonT$};

\end{tikzpicture}
  \end{center}
  \caption{Construction of the synchronous transducer}
  \label{fig:r1Rr2ReverseRecognition}
\end{figure}

The following lemma formally describes how the previous results can be 
used to compute a synchronous resource transducer for the resource-cost 
reachability relation. The proof is divided into three parts. First, we
use the above described ideas to transform an \RPRS{} into an equivalent
profile annotated prefix replacement system. Second, 
Algorithm~\ref{algo:GTTSaturation} is used to compute the annotation-aware
transitive closure. Third, we construct a synchronous resource transducer
from the result of the saturation procedure. The idea of this construction is
depicted in Figure~\ref{fig:r1Rr2ReverseRecognition}. It is based on a 
product construction of the two output automata from the saturation. The 
construction is extended to read an entire pair of configurations. Both state 
components of the product automaton start in an additional wait-state 
($q_{\mathrm{wait}}$). This wait-state is used to skip the padding in front of 
the shorter configuration. Following, both state components process the 
changed prefix of the two configurations individually. To do so, each 
component has its own copy of every counter of the \RPRS{}. The detailed 
analysis shows that this is sufficient to simulate the counters up to a factor
of two. When the changed prefix of the two configurations is completely read,
the constructed transducer nondeterministically guesses that the 
common postfix starts. Whenever both state components are in the same state 
(one of the states of $Q_{\mathrm{shared}}$ in the algorithm), the transducer
can change to the newly introduced state  $q_{\mathrm{id}}$. It then only 
verifies until the end that the remainder of both configurations is identical.

In the following, we only write $\iOp$ but mean the counter operation $\icOp$
of B-automata. This reduces the notation overhead and it resembles more closely
the notation in \RPRS{} where all increments count for the resource-cost value
and there is no notion of \emph{check}.

\begin{lem}
	Let $\prsR = (\Sigma,\DeltaR,\Gamma)$ be a resource prefix replacement
system and $\alpha(k) = 2k + 1$ be a correction function. There is a synchronous
resource transducer $\automatonT$ such that for all $w_1, w_2 \in \Sigma^*$ we have 
\[
\semantics{\automatonT}_{\padprodR}((w_1,w_2)) \costEquiv 
\inf \{ j \in \nat \mid  w_1 \configstepsLesserCost{j} w_2 \}
\]
\end{lem}
\begin{proof}
 The proof follows the previously presented ideas. We translate $\prsR$ into a
monoid annotated prefix replacement system with counter profile annotation. Let $\Gamma = \{\gamma_1, \ldots , \gamma_n \}$ such that the $i$-th component of the profile vector represents the counter $\gamma_i$. 
Subsequently, we use Algorithm~\ref{algo:GTTSaturation} to obtain the automata
$\automatonA_1^* = (Q_1, \Sigma, \{q_{1,\eps}\}, \emptyset,\Delta_1)$ and
$\automatonA_2^* = (Q_2, \Sigma, \{q_{2,\eps}\}, \emptyset, \Delta_2)$. We
define the transducer by: 
\[
	\automatonT = (Q, \alphVector{\Sigma}{2}, \Delta, \In, \Fin, \Gamma')
\]
where the state set is defined as $Q = (Q_1 \dotcup \{q_{\mathrm{wait}}\})
\times (Q_2 \dotcup \{q_{\mathrm{wait}}\}) \dotcup \{q_{\mathrm{id}}\}$ with
initial state $\In = \{(q_{\mathrm{wait}},q_{\mathrm{wait}})\}$ and final state
$\Fin = \{q_{\mathrm{id}}\}$. The set of counters contains two copies for every
counter in $\Gamma$, i.e., $\Gamma' = \{ \gamma^1, \gamma^2 \mid \gamma \in
\Gamma\}$. The transition relation $\Delta$ is given similar to a product
construction but also contains transitions which enable the automaton to wait
until the padding ends and to change from a state where both components of the
state are in the same shared end state to the state which starts reading the
identity. 
\begin{align*}
	\Delta &= \{ ((p_1,p_2),\twovec{a_1}{a_2},(p_1',p_2'),\fraku) \mid (p_l,a_l,p'_l,m_l) \in \Delta_l, l \in \{1,2\},\fraku(\gamma^j_i) := \profileToOp((m_j)_i) \} \\
	      &\cup \{((p,q_{\mathrm{wait}}),\twovec{a}{\$},(p',q_{\mathrm{wait}}),\fraku) \mid (p,a,p',m_1) \in \Delta_1, \fraku(\gamma^2_i) := \eps, \fraku(\gamma^1_i) := \profileToOp((m_1)_i)\} \\
	      &\cup \{((q_{\mathrm{wait}},q),\twovec{\$}{b},(q_{\mathrm{wait}},q'),\fraku) \mid (q,b,q',m_2) \in \Delta_2, \fraku(\gamma^1_i) := \eps, \fraku(\gamma^2_i) := \profileToOp((m_2)_i)\} \\
	      &\cup \{((q_{\mathrm{wait}},p_2),\eps,(q_{1,\eps},p_2),\fraku_{\nOp}),((p_1,q_{\mathrm{wait}}),\eps,(p_1,q_{2,\eps}),\fraku_{\nOp}) \mid p_i \in Q_i \cup \{q_{\mathrm{wait}}\} \} \\
	      &\cup \{ ((p,p),\eps,q_{\mathrm{id}},\fraku_{\nOp}), (q_{\mathrm{id}},\twovec{a}{a},q_{\mathrm{id}},\fraku_{\nOp}) \mid p \in Q_1 \cap Q_2, a \in \Sigma\}
\end{align*}
with $\fraku_{\nOp}(\gamma^j_i) := \eps$

We show the claim by first showing $\semantics{\automatonT}_{\padprodR}((w_1,w_2)) \costlea \inf \{ j \in \nat \mid  w_1 \configstepsLesserCost{j} w_2 \}$ and then $\inf \{ j \in \nat \mid  w_1 \configstepsLesserCost{j} w_2 \} \costlea \semantics{\automatonT}_{\padprodR}((w_1,w_2))$ separately.

Assume $w_1 \configstepsLesserCost{k} w_2$. Then there is a counter profile vector $m$ in which no entry exceeds $k$ and we have $w_1 \configstepsCost{m} w_2$. Note that if no such $k$ exists, there is nothing to show in this direction because the $\inf$ is $\infty$ in this case. By Lemma~\ref{lem:CorrectnessOfGTTSaturation} part (i) there are runs $\automatonA^*_1: q_{1,\eps} \apathS[\rl]{w_1'} q_{(u,v,\bar m)}$ and $\automatonA^*_2: q_{2,\eps}
\apathS[\rr]{w_2'} q_{(u,v,\bar m)}$ for some $(u,v,\bar m) \in \DeltaR$ and $x \in \Sigma^*$ such that $w_1 = w_1'x$, $w_2 = w_2'x$  and $\rl \bar m \rev(\rr)
\le m$. By these two runs, we can easily obtain an accepting run of $\automatonT$ on $w_1 \padprodR w_2$ where the counters $\gamma^1_i$ process a counter sequence with profile $\rl$ and the counters $\gamma^2_i$ process a counter sequence with profile $\rr$. 

We show that the maximal occurring counter value in these runs can only be smaller than the maximum in $\rl\bar m \rev(\rr)$. Consider the sequence of counter operations for one counter. The part resulting from $\rl$ looks like $[\rOp]\iOp^{n^\leftarrow_1}\rOp\iOp^{n^\leftarrow_2}\ldots\rOp\iOp^{n^\leftarrow_{\ell_1}}$, the part resulting from $\rr$ like $[\rOp]\iOp^{n^\rightarrow_{1}}\rOp\iOp^{n^\rightarrow_2}\ldots\rOp\iOp^{n^\rightarrow_{\ell_2}}$. The operation corresponding to $\bar m$ is either $\iOp,\rOp$ or $\nOp$. We first recognize that the maximal counter value only depends on the length of the maximal block of increments in the sequence. Consequently, it makes no difference to read the sequence reversed. The only difference of the sequences read by $\automatonT$ to the combined sequence $\rl \bar m \rev(\rr)$ is that the latter could contain the block $\iOp^{n^\leftarrow_{\ell_1}}\iOp^{\chi_m}\iOp^{n^\rightarrow_{\ell_2}}$ if the operation of $\bar m$ is not $\rOp$ with $\chi_m = 1$ if the operation is $\iOp$ and $0$ otherwise. This block is not read by the run of $\automatonT$. However, this block would only induce a larger value. This shows that the value computed by $\automatonT$ is a lower bound on the $k$ such that $w_1 \configstepsLesserCost{k} w_2$. Thus, we have $\semantics{\automatonT}_{\padprodR}((u,v)) \costlea \inf \{ j \in \nat \mid  w_1 \configstepsLesserCost{j} w_2 \}$.

For the other inequality let now $\semantics{\automatonT}_{\padprodR}((w_1,w_2)) = k$. From the run witnessing this fact, we obtain a common suffix $x \in \Sigma^*$, which is read while the automaton is in state $q_{\mathrm{id}}$, and $w_1',w_2'$ such that $w_1 = w_1'x$, $w_2 = w_2'x$. Furthermore, there is a common state $q_{(u,v,\bar m)} \in Q_1 \cap Q_2$ which is the last state before the automaton changed to $q_{\mathrm{id}}$. By splitting the run on $\automatonT$ at the transition changing to $q_{\mathrm{id}}$ and projecting the first part to the respective state components, we obtain the runs $\automatonA^*_1: q_{1,\eps} \apathS[\rr]{w_1'} q_{(u,v,\bar m)}$ and $\automatonA^*_2: q_{2,\eps} \apathS[\rl]{w_2'} q_{(u,v,\bar m)}$ where the counter profile $\rl$ corresponds to the sequence of counter operations of the counters $\gamma^1_i$ and the profile $\rr$ corresponds to the sequence of the counters $\gamma^2_i$, respectively. By Lemma~\ref{lem:CorrectnessOfGTTSaturation}, we obtain that $w_1 \configstepsCost{\rl \bar m \rev(\rr)} w_2$. Thus, it remains to show that no value in the profile $\rl \bar m \rev(\rr)$ exceeds $2k + 1$. For the sake of simplicity we assume that there is only one counter. The proof extends directly to several counters by repeating all arguments component-wise for every counter. We use equivalent counter sequences to examine the maximal counter value induced by $\rl \bar m \rev(\rr)$. Let $\profileToOp(\rl) = \iOp^{\overleftarrow{\ipl}}[\rOp\iOp^{\overleftarrow{\cmax}}][\rOp\iOp^{\overleftarrow{\ipr}}]$ and $\profileToOp(\rr) = \iOp^{\overrightarrow{\ipl}}[\rOp\iOp^{\overrightarrow{\cmax}}][\rOp\iOp^{\overrightarrow{\ipr}}]$ be the counter sequence representations of the profiles and $\mathsf{op} = \profileToOp(\bar m)$. With this notation, the complete sequence $\rl \bar m \rev(\rr)$ corresponds to $\iOp^{\overleftarrow{\ipl}}[\rOp\iOp^{\overleftarrow{\cmax}}][\rOp\iOp^{\overleftarrow{\ipr}}] \mathsf{op} [\iOp^{\overrightarrow{\ipr}}\rOp][\iOp^{\overrightarrow{\cmax}}\rOp]\iOp^{\overrightarrow{\ipl}}$. First, we remark that $\overrightarrow{\cmax}$ and $\overleftarrow{\cmax}$ are either undefined or at most $k$ by assumption. Since they are enclosed by $\rOp$ in the sequence, we can ignore them. Second, we see that the  maximal counter value of the complete counter sequence only exceeds the maximal counter value of the two single sequences if $\mathsf{op} \ne \rOp$. It becomes largest if $\mathsf{op} = \iOp$. Consequently, the new block of increments in the middle of the sequence is $\overleftarrow{\ipl} + \overrightarrow{\ipl} + 1$, $\overleftarrow{\ipr} + \overrightarrow{\ipr} + 1$, $\overleftarrow{\ipl} + \overrightarrow{\ipr} + 1$ or $\overleftarrow{\ipr} + \overrightarrow{\ipl} + 1$. Since all profile entries are bounded by $k$, the counters are bounded by $2k +1$ in the complete sequence. So, we have that $w_1 \configstepsLesserCost{2k +1} w_2$. In total, we have $\inf \{ j \in \nat \mid  w_1 \configstepsLesserCost{j} w_2 \} \costlea \semantics{\automatonT}_{\padprodR}((w_1,w_2))$.
\end{proof}

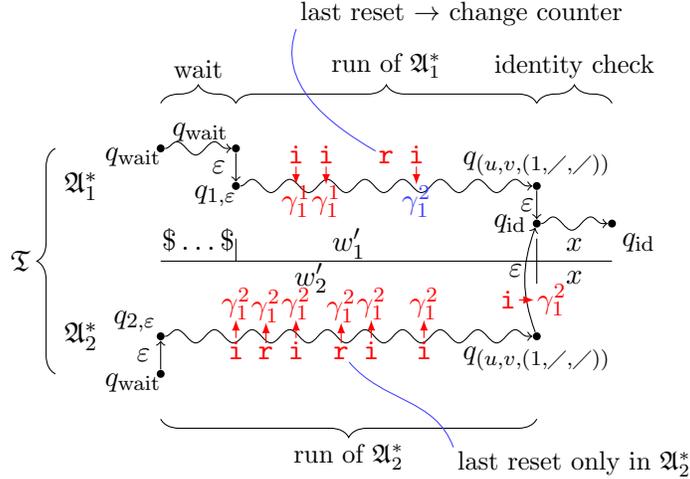
\begin{figure}[t]
  \begin{center}
  \begin{tikzpicture}
  \node[dotstyle] (A2Start) at (0,-0.5) {};
  \node[dotstyle] (A2Start2) at (0,0) {};
  \node[dotstyle] (A1Start) at (0,2.5) {};
  \node[dotstyle] (A1WaitEnd) at (1,2.5) {};
  \node[dotstyle] (A1WaitEnd2) at (1,2) {};
  \node[dotstyle] (A1NormEnd) at (5,2) {};
  \node[dotstyle] (A2NormEnd) at (5,0) {};

  \node[dotstyle] (CommonStart) at (5,1.5) {};
  \node[dotstyle] (CommonEnd) at (6,1.5) {};

  \node[above] at (A1NormEnd) {$q_{(u,v,(1,\na,\na))}$};
  \node[below] at (A2NormEnd) {$q_{(u,v,(1,\na,\na))}$};
  
  \node[below right] at (CommonEnd) {$q_{\mathrm{id}}$};
  \node[left] at (CommonStart) {$q_{\mathrm{id}}$};

  \draw[automatapath] (A1Start) -- (A1WaitEnd);
  \draw[curlybracket] ($(A1Start) + (0,0.6)$) -- ($(A1WaitEnd) + (0,0.6)$) node[midway,above=2mm,font={\small}] {wait}; 
  \draw[automatapath] (A1WaitEnd2) -- (A1NormEnd);
  \draw[curlybracket] ($(A1WaitEnd) + (0,0.6)$) -- ($(A1NormEnd) + (0,1.1)$) node[midway,above=2mm,font={\small}] {run of $\automatonA^*_1$}; 
 
  \draw[curlybracket] ($(CommonStart) + (0,1.6)$) -- ($(CommonEnd) + (0,1.6)$) node[midway,above=2mm,font={\small}] {identity check};

  \draw[automatapath] (A2Start2) -- (A2NormEnd);
  \draw[curlybracket] ($(A2NormEnd) + (0,-1.1)$) -- ($(A2Start) + (0,-0.6)$) node[midway,below=2mm,font={\small}] {run of $\automatonA^*_2$}; 

  \draw[automatapath] (CommonStart) -- (CommonEnd);

  \draw[->] (A2NormEnd) to[out=110,in=-110] 
		node [pos=0.6,left=-1mm] {$\eps$} 
		node[pos=0.3,left=1mm,red,inner sep=1pt] (iEnd) {$\iOp$}  
		node[pos=0.3,right=1mm,red,inner sep=1pt] (gEnd) {$\gamma^2_1$}
			(CommonStart);
  \draw[->,>=latex,red] (iEnd) -- (gEnd);
  
  \draw[->] (A1NormEnd) -- node [midway,left=-1mm] {$\eps$} (CommonStart);

  \node[font={\small},inner sep=0.5mm] (lastReset1Label) at (4,4.3) {last reset $\rightarrow$ change counter};
  \node[font={\small},inner sep=0.5mm] (lastReset2Label) at (5.5,-1.7) {last reset only in $\automatonA^*_2$};

  \draw[-] (0,1) -- (6,1);
  \draw[-] (5,0.7) -- (5,1.3);
  \node[above] at (5.5,1) {$x$};
  \node[below] at (5.5,1) {$x$};

  \draw[-] (1,1) -- (1,1.3);
  \node[above] at (0.5,1) {$\$ \ldots \$$};
  \node[above] at (2.5,0.9) {$w_1'$};
  \node[below] at (2,1.11) {$w_2'$};
  
  \begin{scope}[>=latex,anchor=base]
	\node[color=red,inner sep=1pt] (i11) at (1.8,2.3) {$\iOp$};
	\node[color=red,below=4mm,inner sep=0pt] (g11) at (i11)  {$\gamma^1_1$};
	\draw[->,color=red] (i11) -- (g11);
	
	\node[color=red,inner sep=1pt] (i12) at (2.2,2.3) {$\iOp$};
	\node[color=red,below=4mm,inner sep=0pt] (g12) at (i12)  {$\gamma^1_1$};
	\draw[->,color=red] (i12) -- (g12);
	
	\node[color=red,inner sep=1pt] (r13) at (3,2.3) {$\rOp$};
	\draw[-,color=blue!75] (r13) to[out=150,in=-110] (lastReset1Label.south west);
	
	
	\node[color=red,inner sep=1pt] (i14) at (3.4,2.3) {$\iOp$};
	\node[color=blue!75,below=4mm,inner sep=0pt] (g14) at (i14)  {$\gamma^2_1$};
	\draw[->,color=red] (i14) -- (g14);

	\node[color=red,inner sep=1pt] (i21) at (1,-0.3) {$\iOp$};
	\node[color=red,above=4mm,inner sep=0pt] (g21) at (i21)  {$\gamma^2_1$};
	\draw[->,color=red] (i21) -- (g21);
	\node[color=red,inner sep=1pt] (r22) at (1.4,-0.3) {$\rOp$};
	\node[color=red,above=4mm,inner sep=0pt] (g22) at (r22)  {$\gamma^2_1$};
	\draw[->,color=red] (r22) -- (g22);
	\node[color=red,inner sep=1pt] (i23) at (1.8,-0.3) {$\iOp$};
	\node[color=red,above=4mm,inner sep=0pt] (g23) at (i23)  {$\gamma^2_1$};
	\draw[->,color=red] (i23) -- (g23);
	
	\node[color=red,inner sep=1pt] (r24) at (2.4,-0.3) {$\rOp$};
	\node[color=red,above=4mm,inner sep=0pt] (g24) at (r24)  {$\gamma^2_1$};
	\draw[->,color=red] (r24) -- (g24);
	\draw[-,color=blue!75] (r24) to[out=-50,in=110] (lastReset2Label.north west);
	
	\node[color=red,inner sep=1pt] (i25) at (2.8,-0.3) {$\iOp$};
	\node[color=red,above=4mm,inner sep=0pt] (g25) at (i25)  {$\gamma^2_1$};
	\draw[->,color=red] (i25) -- (g25);
	
	\node[color=red,inner sep=1pt] (i26) at (3.5,-0.3) {$\iOp$};
	\node[color=red,above=4mm,inner sep=0pt] (g26) at (i26)  {$\gamma^2_1$};
	\draw[->,color=red] (i26) -- (g26);
  \end{scope}
  
  \node[left=0.7cm] (A1) at (0,2) {$\automatonA^*_1$};
  \node[left=0.7cm] (A2) at (0,0) {$\automatonA^*_2$};
  \coordinate[left=1.4cm] (Upper) at (A1Start);
  \coordinate[left=1.4cm] (Lower) at (A2Start);
  
  \node[below left=-2mm] at (A1Start) {$q_{\mathrm{wait}}$};
  \node[above left=-0.5mm] at (A1WaitEnd) {$q_{\mathrm{wait}}$};
  \node[below left=-2mm] at (A2Start) {$q_{\mathrm{wait}}$};
  
  \draw[->] (A2Start) -- (A2Start2) node[midway,left] {$\eps$};
  \node[above left=-1mm] at (A2Start2) {$q_{2,\eps}$};
  
  \draw[->] (A1WaitEnd) -- (A1WaitEnd2) node[midway,left] {$\eps$};
  \node[below left=-2mm] at (A1WaitEnd2) {$q_{1,\eps}$};

  \draw[curlybracket] (Lower) -- (Upper) node[midway,left=2mm] {$\automatonT$};

\end{tikzpicture}
  \end{center}
  \caption{Construction of the synchronous transducer for the exact value}
  \label{fig:r1Rr2ReverseRecognitionExact}
\end{figure}

We remark that it is also possible to use the result of the saturation procedure to construct a transducer which calculates the exact values. However, this requires a slightly more complex construction. The basic construction is analogous to the construction presented in the previous lemma. However, we additionally have to solve the problem of the connection point in the middle of the two sequences. This can be done by nondeterministically guessing the position of the last reset for each counter in both components of the state space. This is illustrated in Figure~\ref{fig:r1Rr2ReverseRecognitionExact}. At some point in the run, the transducer $\automatonT$ decides that this is the last reset for some counter in this component of the state space. It checks whether the corresponding counter in the other state-component is already in this \emph{after last reset} mode. If this is the case, the current component uses from this time onwards in the run the counter of the other component to store increments. In any case the automaton validates that there are no more resets for the respective counter in the run. That includes the counter operation associated with the shared synchronization state before switching to $q_{\mathrm{id}}$. This way, the transducer accumulates all increments contained in the middle part connecting both counter sequences in one counter. Hence, the transducer calculates exactly the maximal counter value of $\rl \bar m \rev(\rr)$. 

With the previous lemma, we can directly obtain the following theorem.

\begin{thm}\label{thm:RPRSareResourceAutomatic}
	Let $\prsR = (\Sigma,\Delta,\Gamma)$ be an \RPRS{} and $\resRep{\prsR} =
   (\Sigma^*,\pathTo{}^{\resRep{\prsR}})$ its resource structure representation.
   The structure $\resRep{\prsR}$ is resource automatic.
\end{thm}

We are now ready to prove the main result about the bounded reachability
problem of \RPRS{}.  We already saw in the previous section that we can express bounded
reachability with \FORR{} and that we are able to compute the value of this
logic on resource automatic structures. The previous result completes the
puzzle and enables a concise proof for a positive result in a restricted case of
bounded reachability.

\begin{thm}\label{thm:BoundedReachabilityIsDecidable}
	The bounded reachability problem for resource prefix replacement systems 
    and regular sets $A$ and $B$ of configurations is decidable.
\end{thm}
\begin{proof}
	By Theorem~\ref{thm:RPRSareResourceAutomatic} the structure
$\resRep{\prsR} = (\Sigma^*,\pathTo{}^{\resRep{\prsR}})$ is resource automatic for a given resource
prefix replacement system $\prsR = (\Sigma,\Delta,\Gamma)$. For regular sets
$A, B \subseteq \Sigma^*$ the extended resource representation $\resRep{\prsR}'
= (\Sigma^*, \pathTo{}^{\resRep{\prsR}}, \overline{A}, B)$ is also resource automatic since the
valuations of $\overline{A}$ and $B$ are characteristic functions of a
regular language. 
	By Proposition \ref{prop:BoundedReachabilityIsFORRExpressible} the bounded
reachability problem is expressible in the logic \FORR{} over the structure
$\resRep{\prsR}'$ and by
Theorem~\ref{thm:FORRisEffectivelyComputableOnRAStructures} the semantics of
\FORR{} formula is effectively computable on resource automatic structures.
Consequently, computing the semantics yields a decision procedure for the
bounded reachability problem.
\end{proof}

We remark that the complete saturation process can be (computationally) accelerated
if one is only interested in boundedness. In this case it is sufficient to have counter
profiles with entries from $\{0,1,\na\}$ and ignore higher counter values. 
If larger values are created with concatenation, we just reset them to $1$. 
One can easily verify that such profiles also form a well-partially ordered 
monoid with involution and that the usage of the restricted profiles only 
results in smaller values for the calculated resource-cost. 
However, since there are only finitely many profiles
in the saturated automata, there is a largest value $k$ occurring in the 
profiles. A run on the original automata can, compared to a run on the 
automata only using the entries $\{0,1,\na\}$, only have larger 
resource-cost by the factor $k$. Hence, the computed functions of both 
automata are $\costEquiv[]$-equivalent and this equivalence is known to 
preserve boundedness. Moreover, with this restriction the runtime of the 
saturation procedure is polynomial in the number and maximal length of the 
prefix replacement rules  and exponential in the number of counters 
(since there are still exponentially
many incomparable counter profile vectors with the restricted profile entries). 

We additionally remark that Theorem~\ref{thm:BoundedReachabilityIsDecidable}
can be generalized to prefix replacement systems with labeled replacement rules and regular constraints on
the labeling of paths. One can extend (resource) prefix replacement systems by
labeling all replacement rules with symbols from a finite alphabet $\Lambda$. A
configuration sequence in such a system generates a word $w \in \Lambda^*$. In
these systems, one can study the reachability relation with respect to some
language $L \subseteq \Lambda^*$. A pair of configurations $u$, $v$ satisfies
this relation $u \configstepsLesserCostLangRestrict{L}{k} v$ if $v$ is reachable
from $u$ with a sequence of configurations which yields a word $w \in L$ and
needs at most $k$ resources. One can easily see that boundedness
questions on such systems can be solved with the above framework if $L$ is
regular. One can either simulate an automaton recognizing $L$ on top of the
stack of the prefix replacement system or (if one does not want to change the
stack alphabet) include a finite monoid recognizing $L$ (see, e.g.,
\cite{Sakarovitch}) in the annotation monoid. 

\section{The Bounded Reachability Problem and
cost-WMSO}\label{sec:BoundedReachabilityAndCostWMSO}

Together with the model of B-automata different forms of quantitative logics
have been introduced. The logic \FORR{} presented in this work can also be seen as such a logic. However, directly in connection
with the introduction of B-automata (see \cite{regularcostfunctions}) T.\
Colcombet already considered two dual forms of \emph{cost monadic second-order} logic
(for short cost-MSO) over finite words. Recently, M.\ Vanden Boom showed in
\cite{costwMSO} that the boundedness problem for cost-weakMSO is decidable over
the infinite tree (in weak MSO set quantification only ranges over finite
sets). In this section, we show how the bounded reachability problem for systems 
with one counter can be encoded to cost-weakMSO. This reduction yields,
together with the decidability result by M.\ Vanden Boom, an alternative decision
procedure for a restricted version of bounded reachability problem as solved
with our framework in Theorem~\ref{thm:BoundedReachabilityIsDecidable}. 
Although this reduction does not provide any new decidability result, we 
present it to demonstrate that the logic cost-(weak)MSO, which was initially designed
with the goal of creating an equivalent formalism to B-automata, is capable 
of expressing a natural problem on quantitative systems. Moreover, we think that
it motivates research in comparing the expressive power of the different 
recently introduced quantitative logics.

The logic cost-(weak)MSO extends standard monadic second-order logic with
a bounded existential quantifier $\exists X . |X| \le N$ (or $\ge N$
in the dual form) which is only allowed to appear positively in the
formula. The semantics of this logic is quantitative. For a given
structure each formula is mapped to the minimal (or maximal in the dual form)
$N_0 \in \nat$ such
that the formula is satisfied as normal (weak)MSO formula with a
cardinality restriction of $N_0$ in all bounded existential
quantifications.

For example, let $P_a$ be the proposition indicating all positions at which
there is an $a$ in the word. The formula $\exists X . |X| \le N \;\forall x\;
P_a(x) \rightarrow X(x)$ counts the number of $a$s in a word. The formula
enforces that all positions with an $a$ also have to be in $X$. Consequently,
the set $X$ must have at least the number of $a$s as size in order to satisfy the formula.

The translation of the bounded reachability problem into a cost-weakMSO
formula over the infinite binary tree is conducted in four steps. We
present the approach for a single counter. First, we shift the counter operation
annotation from the prefix replacement rules to the configurations by
considering the incidence graph. In the incidence graph, every transition is
replaced by an additional node. This additional node is then connected to the 
nodes incident to the edge. This graph is still the configuration graph of
a prefix replacement system. Moreover, we can assume without loss of
generality that this replacement system works over the alphabet $\{0,1\}$.
Second, we embed the graph into the binary tree. As usual, we identify every
node in the tree with the $\{0,1\}$-word induced by the path from the root to
the node. In this encoding, the root is labeled with $\eps$, the left child
with $0$ and the right child with $1$. This can be extended to the complete
tree. Hence, every node can represent one configuration of the prefix
replacement system. Furthermore, we can obtain FO-formulas
$\varphi_{\configstep}(a,b)$, representing the successor
relation (see \cite{thomas03a}), $\varphi_\iOp(a),\varphi_\rOp(a)$,
identifying the configurations with increment or reset counter actions, and
$\varphi_A(a), \varphi_B(b)$, marking the tree representation of the
sets $A$ and $B$ of the bounded reachability problem. Third, we
construct a formula that expresses reachability with a bounded number
of increments. Finally, we use this previous formula to construct a
formula for bounded reachability. This formula checks bounded
reachability by guessing positions with resets and checking for a
bounded number of increments between those positions.

\begin{lem}\label{lem:WMSOReachability}
  Let $\theta(a,b)$ be a formula representing a binary relation. The weakMSO
formula $\Psi_{\theta}(P,x,y)$ defined as follows is true iff the finite set $P$
contains a sequence from $x$ to $y$ connected via $\theta$.
  \begin{align*}
	\Psi_{\theta}(P,x,y) := &P(x) \wedge P(y) \wedge \forall S \Big( (S \subseteq
P \wedge S(x))\rightarrow \\ 
		&(S(y) \vee (\exists z \exists z' S(z) \wedge
\neg S(z') \wedge P(z') \wedge \theta(z,z')))\Big)
  \end{align*}
  The following cost-WMSO formula $\Xi(a,b)$ has a value of at most $k$ iff
there is a
path with at most $k$ increments from $a$ to $b$

  $$ \Xi(a,b) := \exists X. |X| \le N\; \exists P\;
\Psi_{\varphi_{\configstep}}(P,a,b) \wedge \forall x (P(x) \wedge
\varphi_{\iOp}(x)) \rightarrow X(x) $$
\end{lem}

\begin{proof}
 We start with the claim on the formula $\Psi_{\vartheta}(P,x,y)$.  Fix
a set $P$ and two elements $x,y$.

First, assume that $P$ contains a $\vartheta$-path $x = u_1,\ldots,u_n = y$ with
$\vartheta(u_i,u_{i+1})$.  Then $P(x)$ and $P(y)$ are satisfied. Now, let $S
\subseteq P$ be an arbitrary subset of $P$ which contains $x$ but does not
contain $y$ (otherwise the second part of the formula is trivially satisfied).
There is a maximal index $i$ such that $u_i \in S$. Since $y \not\in S$, we get
$i < n$ and thus $u_{i+1} \in P \setminus S$. By construction,
$\vartheta(u_i,u_{i+1})$ is satisfied. Thus, $\Psi_{\vartheta}(P,x,y)$ is
satisfied.

Conversely, assume that $\Psi_{\vartheta}(P,x,y)$ is satisfied. Let $S$ be the
set of elements which are $\vartheta$-reachable from $x$ using only elements of
$P$. Note that $S$ is finite and thus occurs in the universal quantification of
the second part of the formula. Furthermore, we have $x \in S$ because $x$ is 
trivially reachable from $x$. If $y \in S$, we get the desired $\vartheta$-path from $x$ to $y$.
Otherwise, by the second part of the formula, there are two elements $z \in S
\subseteq P$ and $z' \in P \setminus S$ such that $\vartheta(z,z')$. By
assumption, $z$ is reachable from $x$ using only elements of $P$. Consequently,
$z' \in P \setminus S$ is also reachable from $x$ using only elements of $P$.
This is a contradiction to the choice of $S$.

We can now prove the claim on the second formula $\Xi(a,b)$. Assume there is a
path with at most $k$ increments from $a$ to $b$. Let $P$ be the set of the
nodes on this path. By the previous proof, we know that
$\Psi_{\varphi_{\configstep}}(P,a,b)$ is satisfied. Furthermore, we set $X = \{
u \in P \mid \varphi_{\iOp}(u) \text{ is true}\}$. By construction, we have $|X|
\le k$. Consequently, the formula $\Xi(a,b)$ has a value of at most $k$.

Assume the formula $\Xi(a,b)$ has a value of at most $k$. Let $X$ and $P$ be
sets witnessing the satisfiability of $\Xi$ with a value of at most $k$.
Furthermore, let $P_\iOp = \{ u \in P \mid \varphi_{\iOp}(u) \text{ is true}\}$.
By definition, we have $|X| \le k$. By the second part of the formula, we obtain
that $P_\iOp \subseteq X$ and thus $|P_\iOp| \le k$. Since
$\Psi_{\varphi_{\configstep}}(P,a,b)$ guarantees that $b$ is reachable from $a$
through only configurations in $P$, we obtain that there is a path from $a$ to
$b$ with at most $k$ increment configurations.
\end{proof}

\begin{thm}\label{thm:WMSOBoundedReachability}
  The set $B$ is boundedly reachable (in a one counter setting) from $A$ iff the following cost-weakMSO
formula, which uses the previously defined reachability formula $\Psi$ with 
respect to the bounded-cost reachability formula $\Xi$, has a finite value
  $$ \forall a \varphi_A(a) \rightarrow (\exists b \varphi_B(b) \wedge \exists
R\; \Psi_{\Xi}(R,a,b) \wedge \forall r R(r) \rightarrow (\varphi_\rOp(r) \vee r
= a \vee r = b))$$
\end{thm}

\begin{proof}
Let the value of the formula in the theorem be less than $k$. This means that
for all $a \in A$ the second part of the formula (after the implication) has a
value less than $k$. Let $b_a$ and $R_a$ be witnesses for this. From the last
part of the formula, we obtain that $\varphi_\rOp$ is true for all elements in
$R_a \setminus \{a,b_a\}$. Since $\Psi_\Xi(R_a,a,b_a)$ must have a value of less
than $k$, there is a sequence $a = r_1,\ldots,r_m = b_a$ such that
$\Xi(r_i,r_{i+1})$ has value less than $k$. By the previous lemma, we obtain
that there is a path from $r_i$ to $r_{i+1}$ with less than $k$ increments.
Since the positions $r_i$ are reset positions (except $r_1$ and $r_m$), this
path is a witness that $a$ reaches an element $b_a \in B$ with cost less than
$k$. Hence, $B$ is boundedly reachable from $A$.
  
Conversely, let $B$ be boundedly reachable from $A$ with a resource bound of
$k$. We show that the value of the formula is at most $k$. So, let $a \in A$. By
assumption, there is a $b_a \in B$ and a path $a = u_1, \ldots, u_m = b_a$ with
resource-cost of at most $k$. We set $R_a := \{ u_i \mid \varphi_\rOp(u_i)
\text{ is true}\} \cup \{a,b_a\}$. Since the value of the path is at most $k$,
the segments of the path between two subsequent reset positions have at most $k$
increments. Hence, $\Psi_\Xi(R_a,a,b_a)$ has a value of at most $k$. In total,
the whole formula has a value of at most $k$ because $a$ was chosen arbitrarily.
\end{proof}

The previously presented approach is specially tailored for the 
bounded reachability problem on resource prefix replacement systems with 
only one counter. Although it exemplary shows that one can express 
bounded reachability in a quite straightforward way with cost-(W)MSO,
the approach lacks the flexibility and generality of the previously studied
framework. One can see this already when trying to cover the case of
multiple counters. Although it might be possible to extend the shown approach
to multiple counters, some significant new ideas are required.

\section{Conclusion}
\label{sec:Conclusion}

We introduced \RPRS{} as a model for recursive programs with resource
consumption. This model can represent recursive programs that use discrete
resources in a consume-and-refresh manner. We represented these operations with
non-negative integer counters. These counters can be incremented to represent
the consumption of a single resource and reset to zero to simulate a complete
refreshment. 

We developed resource structures and the logic \FORR{} to analyze systems with
resources and specify combined properties on the systems' behavior and their
resource consumption. We identified the subclass of resource automatic
structures and provided an algorithm to compute the semantics of \FORR{}
formulas. This logic is able to express the bounded reachability
problem on a presentation of an \RPRS{} as resource structure. Thereby, we
reduced bounded reachability to an evaluation problem of the quantitative logic
\FORR{}. 

We devised a method to compute a synchronous transducer recognizing an
annotation aware transitive closure of an annotated prefix replacement system.
We used this method to prove that the resource structure representation of an
\RPRS{} is resource automatic. This completed the decidability proof of the
bounded reachability problem. 

Although we gave a self-contained presentation on a formalization of recursive
programs with resource consumption and the bounded reachability problem,
several open questions remain. First, the choice of prefix replacement systems
as underlying computational model induces certain restrictions on the systems
which can be modeled. It excludes important classes such as systems with
parallelism or reactive systems. It remains open whether and how the presented
results can be extended to such classes of system models. Second, the positive
computability results for \FORR{} on resource automatic structures motivate
further research in the area of specification logics for systems with resource
consumption. Moreover, in the area of automatic program verification usually
temporal logics are used. However, it remains unclear how temporal logics for
systems with resource consumption look and whether there will be
algorithmic solution methods. Nevertheless, we saw in the previous section that
there are other quantitative logics which are related to verification problems
of systems with resources. Especially, there are several logics which emerged
around the models of B- and S-automata. In addition to cost-(weak)MSO, there is the
logic (weak)MSO+U introduced by M. Bojańczyk in \cite{wmsoPu}. It extends normal
MSO by a quantifier stating that there are sets of unbounded size such that the
following part of the formula holds. An investigation of the relations between
all these logics would help to identify the boundaries of decidability.

\nocite{regularcostfunctions}
\bibliographystyle{alpha}
\bibliography{literature}
\end{document}